\documentclass[12pt,reqno]{amsart}

\usepackage[a4paper,left=1.5cm,right=1.5cm]{geometry}
\usepackage{epsfig}
\usepackage{varioref}
\usepackage{graphicx}
\usepackage{amsmath}
\usepackage{eqnarray}
\usepackage{nccmath}
\usepackage{amsfonts}
\usepackage{xcolor}
\usepackage{amsthm}
\usepackage{graphicx}
\usepackage{amssymb}
\usepackage{bibtopic}

\makeatletter
\newcommand{\vast}{\bBigg@{3}}
\newcommand{\Vast}{\bBigg@{5}}
\makeatother

\DeclareMathOperator{\supp}{supp}

\renewcommand{\div}{\operatorname{div}}
\newcommand{\rot}{\operatorname{rot}}

\newcommand{\bP}{{\mathbb{P}}}
\newcommand{\bE}{{\mathbb{E}}}

\def\sqr#1#2{{\vcenter{\vbox{\hrule height .#2pt
     \hbox{\vrule width .#2pt height#1pt \kern#1pt \vrule
     width .#2pt} \hrule height .#2pt}}}}

\def\acapo{$\null$ \par {\vskip  0mm plus .2mm } \noindent }

\def\finedim{{\hfill\hbox{\enspace${ \square}$}} \smallskip}    

\newcommand{\R}{{\mathbb{R}}}
\newcommand{\E}{{\mathbb{E}}}

\newcommand{\C}{{\mathbb{C}}}
\newcommand{\N}{{\mathbb{N}}}

\newcommand{\Ba}{{\mathcal{B}}}
\newcommand{\ba}{{\bf{a}}}
\newcommand{\bB}{{\bf{B}}}
\newcommand{\Hi}{{\mathcal{H}}}
\newcommand{\cL}{{\mathcal{L}}}

\newcommand{\Tr}{{\mathrm{Tr}}} 
\newcommand{\cC}{{\mathcal{C}}}

\newcommand{\Fo}{{\mathcal{F}}}

\newtheorem{teorema}{Theorem}

\newtheorem{Remark}{Remark}

\newtheorem{definition}{Definition}

\newtheorem{corollary}{Corollary}

\newtheorem{lemma}{Lemma}

\theoremstyle{definition} 

\def\DDD{\acapo \hbox{{\sl Proof [of Theorem \ref{teoDyson}].}$\;\; $}} 
\def\finedim{{\hfill\hbox{\enspace${ \Box}$}} \smallskip}

\begin{document}

\title[\scalebox{0.7}{A rigorous mathematical construction of Feynman path integrals for the Schr\"odinger equation with magnetic field}]{A rigorous mathematical construction of Feynman path integrals for the Schr\"odinger equation with magnetic field}

\author[S. Albeverio, N. Cangiotti and S. Mazzucchi]{S. Albeverio$^1$, N. Cangiotti$^2$ and S. Mazzucchi$^2$}
\address{$^1$ Institute of Applied Mathematics, and Hausdorff Center of Mathematics, University of Bonn, Endenicher Allee 60, 53115 Bonn, Germany}
\email{albeverio@iam.uni-bonn.de}
\address{$^2$Department of Mathematics, University of Trento and INFN-TIFPA, via Sommarive 14, 38123 Povo (Trento), Italy}
\email{nicolo.cangiotti@unitn.it}
\email{sonia.mazzucchi@unitn.it}

\begin{abstract}
A Feynman path integral formula for the Schr\"odinger equation with magnetic field is rigorously mathematically realized in terms of infinite dimensional oscillatory integrals. 
We show (by the example of a linear vector potential) that  the requirement of the independence of the integral on the approximation procedure forces the introduction of a counterterm to be added to the classical action functional. This provides a natural explanation for the appearance of a Stratonovich integral in the path integral formula for both the Schr\"odinger and heat equation with magnetic field.
\ \\

\bigskip


\end{abstract}

\maketitle

\vskip 1\baselineskip


\section{Introduction}
\label{Intro}
Since their first introduction in the 1940s \cite{Fey,FeyHib}, Feynman path integrals have been considered  a powerful tool in theoretical physics on one hand  and  a source of mathematical challenges on the other hand. Indeed since the 1960s several efforts have been devoted to the development of an infinite dimensional integration theory providing a rigorous mathematical definition of   Feynman's heuristic formula 
\begin{equation}\label{Fey1}
\psi (t,x)=C^{-1}\int_\Gamma e^{\frac{i}{\hbar}S(\gamma)}\psi_0(0, \gamma (0))d\gamma
\end{equation}
 for the solution to the time dependent
Schr\"odinger equation
\begin{equation}\label{Sch}i\frac{\partial}{\partial t}\psi(t,x)=-\frac{1}{2}\Delta\psi(t,x) +V(x)\psi(t,x),\qquad x\in \R^d, t\in \R,\end{equation}
describing the time evolution of the state $\psi\in L^2(\R^d)$ of a non-relativistic quantum particle moving in the $d-$dimensional Euclidean space under the action of the force field associated to a real valued potential $V$. According to Feynman's proposal, the state of the particle should be given by an heuristic integral of the form \eqref{Fey1}
on the space $\Gamma$ of continuous paths $\gamma :[0,t]\to \R^d$ with fixed end point $\gamma(t)=x$. The integrand in \eqref{Fey1}, namely the function
$$S(\gamma)=\int\cL(\gamma(\tau),\dot\gamma(\tau) )d\tau=\int_0^t\left( \frac{|\dot\gamma(\tau)|^2}{2}-V(\gamma(\tau))\right)d\tau,$$ is the classical action functional  evaluated along the path $\gamma$, where $\cL$ denotes the Lagrangian. Here, $\dot{\gamma}(\tau)$ is the derivative of $\gamma$ at $\tau$, and $|\cdot |$ is the norm in $\R^d$. The symbol  $d\gamma$ stands for a heuristic Lebesgue-type measure on $\Gamma$ and 
\[
C=\int_\Gamma e^{\frac{i}{2\hbar}\int_0^t|\dot\gamma(s)|^2ds}d\gamma,
\] 
plays the role of a normalization constant. 
Even if formula \eqref{Fey1}, as it stands, lacks of a well defined mathematical meaning, it has been widely applied in many areas of quantum physics, providing in fact a quantization procedure and allowing, at least heuristically, to associate a quantum dynamics to any classical Lagrangian. Feynman himself was aware of the lack of a sound mathematical theory for its formula\footnote{Feynman commented this by his often quoted statement {\em``one must feel as Cavalieri must have felt before the invention of the calculus"}.}. In fact, neither the normalization constant $C$, nor the Lebesgue type measure $d\gamma$ are well defined.

In the physical literature, in most cases Feynman's formula is  interpreted  as the limit of a finite dimensional approximation procedure. Indeed, if we divide the time interval $[0,t]$ into $n$ equal parts of amplitude $t/n$, and if for any path $\gamma:[0,t]\to\R$ we consider its approximation by means of a broken line path $\gamma_n$ in $\R^d$ defined as:
\begin{equation}
	\gamma_n(s):=x_j+\frac{(x_{j+1}-x_j)}{t/n}(s-jt/n), \ \ s\in \left [\frac{jt}{n},\frac{(j+1)t}{n} \right],
\end{equation}
where $x_j:=\gamma (jt/n)$ and $j=0, \dots, n-1$, formula \eqref{Fey1} can be interpreted as the limit for $n \to \infty$ of the following approximation:
\begin{equation}
\label{trotter2}
(2 \pi  \hbar i)^{-nd/2}\int_{\R^{nd}} \psi_0(x_0) e^{\frac{i}{\hbar}\sum_{j=1}^n\frac{|x_{j+1}-x_{j}|^2}{2t/n}-\frac{i}{\hbar }\sum_{j=1}^nV(x_j)\frac{t}{ n}}dx_0\dots dx_{n-1}.
\end{equation}
In fact, under rather general assumption on the potential $V$ and the initial datum $\psi_0$, the limit for $n\to \infty $ of the sequence of finite dimensional integrals \eqref{trotter2} converges to the solution of \eqref{Sch} (see, e.g., \cite{Nel, Nic16, Sim,JoLa,Tru1,Ichi0} for a discussion of this approach). 

On the other hand, from a mathematical point of view, the construction of an integration theory allowing to realize Feynman's formula in terms of a Lebesgue integral with respect to a well defined ($\sigma$-additive) complex measure on the infinite dimensional space of paths $\Gamma$ presents severe problems  and it is in fact in most cases impossible (see \cite{Cam,Tho,AlMa16, Ma11,RieGro} for a discussion of this issue).    Several approaches have been proposed, relying, e.g., on analytic continuation of Wiener integrals \cite{Cam,Dos}, or  on an  infinite dimensional distribution theory \cite{CarDW-M,HKPS}, or on suitable approximation procedures \cite{AlBr, AlBoBr, ELT,Fu,KuGoFu,Nel,Tru2,FuTs,42b} . We shall focus on the {\em infinite dimensional oscillatory integral approach}, originally proposed by K. It\={o} in the 1960s \cite{Ito1,Ito2} and further developed by S. Albeverio and R. H{\o}egh-Krohn  \cite{AlHK, AlHKMa} in the 1970s.  The main idea is the generalization of the classical theory of oscillatory integrals on finite dimensional vector spaces \cite{Hor,Dui} to the case where the integration is performed on an infinite dimensional real separable Hilbert space \cite{ELT,Ma}. It is important to point out that this approach allows for a systematic implementation of an infinite dimensional version of the stationary phase method and the corresponding application to the study of the semiclassical limit of quantum mechanics \cite{AlHK,Re}, i.e. the analysis of the detailed asymptotic behavior of the solution of the Schr\"odinger equation when the Planck constant $\hbar $ is regarded mathematically as a small parameter allowed to converge to $0$.

In this paper we shall focus on   the Schr\"odinger equation for a non-relativistic quantum particle moving under the influence of a (rather general) magnetic field $\bB$ associated to a vector potential $\ba$
\begin{equation}\label{SchMF}
i\hbar \frac{\partial}{\partial t}\psi(t,x) =\frac{1}{2}\left(-i\hbar\nabla -\lambda\ba(x)\right)^2\psi(t,x), 
\end{equation}
where $\ba(x) \in \R^d$ and $\lambda\in \R$ plays the role of a coupling constant. We shall construct a rigorous mathematical definition for the corresponding Feynman path integral formula
 \begin{equation}
 \label{Sch1}
       \psi (t,x)=\int _{\gamma(t)=x}e^{\frac{i}{2\hbar}\int_0^t|\dot\gamma(s)|^2ds+\frac{i}{\hbar}\int_0^t  \lambda\ba(\gamma(s))\cdot \dot \gamma (s)ds }
      \psi_0(\gamma(0))d\gamma,
  \end{equation}
in terms of infinite dimensional oscillatory integrals.
In the physical literature \cite{Ful,GaSc,GavMiRoSc,Schu2} the problem of the definition of Feynman path integrals in the presence of magnetic field has been extensively investigated. The traditional procedure relies upon a time slicing approximation of the form
\begin{equation}
\label{time-slicing}
\int e^{\frac{i}{2\hbar}\sum_i\frac{|	\gamma(t_{i+1})-\gamma (t_i)|^2}{t_{i+1}-t_i}+\frac{i}{\hbar}\sum_i  \lambda\ba (\gamma(\tilde t_i))\cdot (\gamma(t_{i+1})-\gamma (t_i))}\psi_0(\gamma(0)) \prod_i
\frac{d\gamma(t_i)}{(2\pi i \hbar (t_{i+1}-t_i))^{1/2}},
\end{equation}
where $0=t_0<t_1<t_2<\dots<t_n=t$ and $\tilde t_i\in [t_i,t_{i+1}]$, $i=0,\dots, n-1$. However this procedure presents some ambiguities since different choices of the point  $\tilde t_i\in [t_i,t_{i+1}]$ lead to different results. The correct choice relies on  the so-called mid-point rule which requires that in the formula \eqref{time-slicing} the vector potential $\ba$ is evaluated at the point 	$\tilde t_i\equiv\frac{\gamma(t_{i+1})+\gamma (t_i)}{2}$. In the Euclidean version of Feynman formula, namely the Feynman-Kac-It\=o formula \cite{Sim} for the solution of the corresponding heat equation in a magnetic field, this procedure yields the Wiener integral representation 
\begin{equation}
\label{heat1.0}
u(t,x)=\E \left [u(x+\omega (t)) e^{-i\int_0^t \lambda\ba(\omega (s)+x)\circ d\omega (s)} \right ],
\end{equation}
where $\circ d\omega(s)$ denotes the Stratonovich stochastic integraland $\E$ is the expectation with respect to Wiener measure for the standard Brownian motion. In fact other choices for the point $\tilde t_i\in [t_i,t_{i+1}]$ in Eq. \eqref{time-slicing} would lead to different stochastic integrals but, as pointed in \cite{Sim,Schu2}, the {\em mid point rule}, or, equivalently, the Stratonovich stochastic integration, is the only one yielding the {\em gauge invariance} of formulas \eqref{Sch1} and \eqref{heat1.0}. In fact, this particular approximation  procedure can be obtained by means of Trotter product formula \cite{GavMiRoSc}.
Other approaches were proposed in the work by Z. Haba \cite{Haba}, relying in principle on an analytic continuation of Wiener integrals, and by W. Ichinose  \cite{Ichi1, Ichi2}, based on particular time-slicing approximations. Rigorously defined infinite dimensional oscillatory integrals have been applied in  \cite{AlBrMF} to the case of a constant magnetic field in the Coulomb gauge.

The aim of the paper is twofold. First of all, for general vector potentials $\ba$, we prove that the finite dimensional approximation procedure associated to the definition of infinite dimensional oscillatory integrals provides the correct construction of the Feynman integral for the magnetic field {\em  without any additional prescription}. We shall show that the mid-point rule has not to be postulated but it is a direct consequence of our construction (for a discussion of these issues as well as the  inclusion of magnetic fields in path integral formulas see, e.g., \cite{Fu,AnDri,Tsu,Kolo,LoTh}). 

The second result of the paper concerns the dependence of the Feynman path integral  on the sequence of finite dimensional approximations introduced in the construction. We show that the requirement of the independence of the particular form of the approximation procedure leads to the introduction of a counterterm  in the classical action functional. In the case of a constant magnetic field we provide a formula that is gauge-independent, generalizing a similar result proposed in \cite{AlBrMF} . It is worthwile to recall that the case of the uniform magnetic was also studied in \cite{IkMa} for Wiener path integrals in relation with the Van Vleck-Pauli formula. Let us also mention that the study of heat semigroup with magnetic field via Feynman-Kac-It\^o formula \cite{Sim,Str} has been extended to the case of fractals, e.g. \cite{HiRoTe}, graphs \cite{GuKeSc}, and manifolds \cite{BrHuLe, Gu16,Sunada,Cyc87}.

The paper is organized as follows. In Section \ref{sez2} we present an overview of the theory of infinite dimensional oscillatory integrals  and their applications to Feynman path integration. In Section \ref{sez3.0} we present some functional analytical results on the Schr\"odinger equation \eqref{SchMF} with magnetic field and, under suitable assumptions on the  vector potentials $\ba$ and the initial datum $\psi_0$, prove that the series expansion of the solution in powers of the coupling constant has a finite radius of convergence. In Section \ref{sez3} we study the Schr\"odinger equation \eqref{SchMF} for analytic vector potentials $\ba$ and construct a Feynman path integral representation for its solution in terms of a particular class of infinite dimensional oscillatory integrals. In Section \ref{sez4} we consider the particular case of a constant magnetic field and provide a renormalized Feynman path integral formula which allows to obtain the independence of the construction procedure of the particular choice of finite dimensional approximations. \\ Let us remark that all our results extend to the case where  a potential term $V$ is added on the right hand side of \eqref{SchMF} (see Remarks \ref{RemDysonV}, \ref{RemDysonV-1}, and \ref{RemDysonV-2}).


\section{Infinite dimensional oscillatory integrals and Feynman maps}
\label{sez2}
Finite dimensional oscillatory  integrals are objects of the following form 
\begin{equation}\label{osc1}
\int_{\R^n}e^{\frac{i}{\epsilon}\Phi(x)}f(x)dx,
\end{equation}
where $\epsilon\in \R\setminus \{0\}$ is a real parameter, $\Phi:\R^n\to \R$ and $f:\R^n\to \C$ are Borel functions, $\Phi$ is usually called {\em phase function}. 
Particular example of integrals of this form are the so-called Fresnel integrals, where  $\Phi$ is a quadratic form:
\begin{equation}
\label{fresnel1}
\int_{\R^n}\frac{e^{\frac{i}{2\epsilon}|x|^2}}{(2\pi i\epsilon)^{n/2}}f(x)dx.
\end{equation}
They find applications in optics and in the theory of wave diffraction and, from a purely mathematical point of view, they have been extensively studied in connection with the theory of Fourier integral operators \cite{Hor, Dui, AlMa2005}. 
According to these studies, in particular \cite{Hor}, the integrals \eqref{osc1} can be computed even when the function $f$ is not summable as the limit of a sequence of regularized integrals.
\begin{definition}
\label{def-osc-int-fin-dim}
Let $f:\mathbb{R}^n\to\mathbb{C}$ be a Borel function and $\Phi :\mathbb{R}^n\to\mathbb{R}$ a phase
function.   If for each Schwartz test function $\phi\in \mathcal{S} (\mathbb{R}^n)$ such that $\phi
(0)=1$ the
integrals
$$I_\delta(f,\phi):=\int_{\mathbb{R}^n}e^{i\frac{\Phi(x)}{\epsilon}}f(x)\phi (\delta x)dx$$
exist for all $\epsilon>0$, $\delta>0$ and 
$\lim  _{\delta \to 0}I_\delta(f,\phi)$
exists and is independent of $\phi$, then the limit is
called the oscillatory integral of $f$ with respect to $\Phi$ and denoted
by
\begin{equation}
\label{int-findim1}
\int_{\mathbb{R}^n}^oe^{i\frac{\Phi(x)}{\epsilon}}f(x)dx\equiv I^{\frac{\Phi}{\epsilon}}(f).\end{equation}
\end{definition}
In fact, according to the definition above, the convergence of the oscillatory integral \eqref{osc1} for $f\notin L^1(\R^n)$ can be obtained  by exploiting the cancellations due to the oscillatory behavior of the integrand.
 Moreover, a particular technique of asymptotic analysis, the {\em stationary phase method}, allows the 
 study of their asymptotic behavior in the limit  when the parameter $\epsilon$ converges to 0 \cite{Dui,Hor,Mur}.

The generalization of this integration technique to the case where $\R^n$ is replaced by a real separable infinite dimensional Hilbert space $(\Hi,\langle \; ,\;
\rangle )$ has been introduced in \cite{AlHK, AlHKMa} and further developed in \cite {ELT,AlBr}. An {\em infinite dimensional oscillatory integral}  with quadratic phase function, also called {\em infinite dimensional Fresnel integral}, is defined as the limit of sequences of finite dimensional approximations.
 
\begin{definition}
\label{intoscinf1}
A function $f:\Hi\to\C$ is  said to be Fresnel integrable
if  for any sequence $\{P _n\}_n$ of projectors
onto n-dimensional subspaces of $\Hi$, such that $P _n\leq P_{n+1}$ and $P_n \to \mathbb{I}$ strongly as $ n \to \infty$ ($\mathbb{I}$ being the identity operator in $\Hi$), the  oscillatory integrals
$$  
\int _{P _n\Hi}^oe ^{i\frac{\|P _n x\|^2}{2\epsilon} }f(P _n x )d (P _nx ),
$$
are well defined (in the sense of Def. \ref{def-osc-int-fin-dim}) and the limit
\begin{equation}
\label{infdim}
\lim_{ n \to \infty} (2\pi i \epsilon )^{-n/2}  \int_{P _n\Hi}^oe ^{i\frac{\|P _n x\|^2}{2\epsilon} }f(P _n x )d (P _nx )
\end{equation}
exists and is independent of the sequence $\{ P _n\}_n$. In this case the limit is  called infinite dimensional oscillatory integral of
$f$  and is denoted by $$ \widetilde{\int_\Hi } e
^{i\frac{\| x\|^2}{2\epsilon} }f(x)dx.$$
\end{definition}

A complete intrinsic characterization of the class of Fresnel integrable function constitutes an open problem of harmonic analysis, even in finite dimension.  The following theorem,  proved in \cite{ELT,AlBr}, provides a particular example of Fresnel integrable functions that play an important role in the construction of Feynman path integrals. In the following we shall denote by $\Fo(\Hi)$ the space of functions $f:\Hi\to \C$ that are Fourier transform of complex (bounded) Borel measures on $\Hi$, namely functions of the form:
\begin{equation}
\label{f-mu}
f(x)=\int _\Hi e^{i\langle x, y\rangle}d\mu_f(y), \qquad x\in \Hi,
\end{equation} 
for some complex bounded Borel measure $\mu_f$ of finite total absolute variation $|\mu_f|(\Hi):=\int_{\Hi}d|\mu_f|$ on $\Hi$. In fact the space $\Fo(\Hi)$ is a Banach algebra, where the product is the pointwise one and the norm $\| f \|_{\Fo(\Hi)}$ of a function $f$ is defined as the total variation of the associate measure $\mu_f$ (through \eqref{f-mu}) .

\begin{teorema}
\label{teo-pars-1}
Let $L:\Hi\to \Hi$ be a self adjoint trace class operator, such that $I-L$ is invertible and let  $f\in \Fo(\Hi)$. Then the function $g:\Hi\to \C$ defined by
 \begin{equation}\label{Fresenlg} g(x)=e^{-\frac{i}{2\epsilon}\langle x,Lx\rangle}f(x), \quad x\in \Hi\end{equation}
 is Fresnel integrable and its infinite dimensional Fresnel integral  is given by the following Parseval-type equality:
\begin{equation}
\label{Parseval}
\widetilde{\int_{\Hi}} e^{\frac{i}{2\epsilon}\langle x,( I-L)x\rangle}f(x)dx=(\det (I-L))^{-1/2} \int_{\Hi} e^{-\frac{i\epsilon}{2}\langle x,(I-L)^{-1}x \rangle} d\mu_f(x)
\end{equation}
where $\det (I-L)$ is the Fredholm determinant of the operator $(I-L)$ (that is the product of the eigenvalues of $(I-L)$) and $\mu_f$ is the complex bounded Borel measure on $\Hi$ related to $f$ by \eqref{f-mu}. 
\end{teorema}

A direct consequence of the Parseval-type equality \eqref{Parseval}  is the following estimate, valid for any function $f\in \Fo (\Hi)$
\begin{equation}
\label{estimate}\left|  \widetilde{\int_{\Hi}} e^{\frac{i}{2\epsilon}\langle x,x\rangle}f(x)dx\right|\leq \int_\Hi d|\mu_f|(x)=|\mu_f|(\Hi)=\| f\|_{\Fo (\Hi)}.
\end{equation}
\begin{Remark}\label{fresnelIntegrableELT}
It is interesting to point out that the class of Fresnel integrable functions contains elements different from the  ones described by eq \eqref{Fresenlg}. As remarked in \cite{ELT}, Def. \ref{intoscinf1} allows to handle, e.g.,   unbounded functions $f:\Hi \to\C$ of the form:
\begin{equation}
\label{fresELT}
f(\gamma)=e^{i\langle \gamma, v\rangle } \langle \gamma,w_1\rangle\cdots \langle \gamma,w_n\rangle, \qquad \gamma \in \Hi,
\end{equation}
with $v,w_1, \dots, w_n \in \Hi$.
\end{Remark} 

\begin{Remark}
Theorem \ref{teo-pars-1} has been extended in \cite{AlBrMF} to the case where the operator $L$ is  Hilbert-Schmidt but it is not trace-class.  In \cite{Ma17} the definition of infinite dimensional oscillatory integrals has been generalized to polynomial phase functions and applied to the construction of Feynman-Kac type formulae for the high-order heat-type equations (see also \cite{AlCaMa18}).
\end{Remark} 

The heuristic Feynman path integral representation \eqref{Fey1} for the solution of the Schr\"odinger equation can be rigorously mathematically realized as an infinite dimensional oscillatory integral on a suitable Hilbert space of continuous ``paths''. Indeed, let us set $\epsilon \equiv \hbar$ and let us consider the so-called Cameron-Martin space $\Hi_t$, that is the Hilbert space of of absolutely continuous paths $\gamma :[0,t]\to\R^3$ with $\gamma(0)=0$  and square integrable weak derivative $\int_0^t|\dot\gamma(s)|^2ds<\infty$, endowed with the  inner product
$$\langle \gamma _1,\gamma _2\rangle =\int ^t _0 \dot\gamma _1(s)\cdot\dot\gamma
  _2(s)ds .$$
  In the case where the potential $V$ in the Schr\"odinger equation \eqref{Sch} is the sum of an harmonic oscillator term and a bounded perturbation $v$ belonging to the space $\Fo(\R^3)$ 
  \begin{equation}
  \label{Vx}
  V(x)=\frac{1}{2}\Omega^2 |x|^2+v(x), \qquad x\in \R^3,
  \end{equation}
and if the initial datum $\psi_0$ belongs to $ L^2(\R^3)\cap \Fo(\R^3)$, it has been proved (see, e.g., \cite{AlHKMa, ELT, Ma17}) that the function $f:\Hi\to \C$ given by 
   $$f(\gamma)=e^{-\frac{i}{\hbar}\int_0^t V(\gamma (s)+x)ds}\psi _0(\gamma (0)+x), \qquad x\in \R^3, \gamma \in \Hi_t,$$
   is Fresnel integrable. Further, its infinite dimensional oscillatory integral, namely
   $$\widetilde{\int_{\Hi_t}} e^{\frac{i}{2\hbar}\langle \gamma, \gamma \rangle}f(\gamma)d\gamma\equiv \widetilde{\int_{\Hi_t}} e^{\frac{i}{2\hbar}\int_0^t |\dot\gamma(s)|^2ds}e^{-\frac{i}{\hbar}\int_0^t V(\gamma (s)+x)ds}\psi _0(\gamma (0)+x)d\gamma,$$
provides a representation for the solution  in the (complex) $L^2(\R^3)$ space of the Schr\"odinger equation \eqref{Sch}. 
   \begin{Remark}
   This result, valid also for $\R^3$ replaced by $\R^d$, for any $d \in \N$, has been generalized in \cite{AlBrMF} to the case where $d=2$ and a constant magnetic field is present, providing a formula valid in the Coulomb gauge for the vector potential $\ba(x,y)=(-By/2, Bx/2)$,  $(x,y)\in \R^2$. 
   \end{Remark}
   
In the context of \eqref{Vx} a particular class of finite dimensional approximations plays an important role and has been introduced in an alternative definition of infinite dimensional oscillatory integrals on the Cameron-Martin space. For fixed $n\in \N$, let $H_n\subset \Hi_t$ be the finite dimensional subspace of piecewise linear paths of the form
   \begin{equation}
\label{gamman}   
   \gamma (s)=\sum_{k=1}^n\chi_{[t_{k-1},t_k]}(s)\left(\gamma (t_{k-1})+\frac{\gamma(t_{k})-\gamma(t_{k-1})}{t_{k}-t_{k-1}}(s-t_{k-1})\right), 
   \end{equation}
   where $s\in [0,t]$, $t_{k}=\frac{kt}{n}$ and $k=0,\dots, n$. Let $P_n:\Hi_t\to\Hi_t$ the projector operator onto $H_n$, whose action on a generic vector $\gamma \in \Hi_t$ is given by the right hand side of \eqref{gamman}. In fact the sequence of  operators $\{P_n\}_n$ converges strongly to the identity operator $\mathbb{I}$ as $n\to \infty$ \cite{Tru1}. 
   
    In this context the definition of {\em Feynman map} has been proposed \cite{Tru2}. It is defined as a linear functional $I_F$ whose action on functions $f:\Hi\to\C$ is defined by 
   \begin{equation}\label{feymap}
  I_F(f)= \lim_{n\to\infty} \frac{ \int _{P _n\Hi}^o e ^{i\frac{\|P _n \gamma \|^2}{2\hbar} }f(P _n \gamma )d (P _n\gamma )}{  \int _{P _n\Hi}^o e ^{i\frac{\|P _n \gamma \|^2}{2\hbar} }d (P _n\gamma )},
   \end{equation}
   whenever the limit on the right hand side exists.
    
   In the case where $f\in \Fo(\Hi_t)$ then its Feynman map is well defined and coincides with its infinite dimensional oscillatory integral:
   $$I_F(f)=\widetilde{\int_{\Hi_t}} e^{\frac{i}{2\hbar}\langle \gamma, \gamma \rangle}f(\gamma)d\gamma.$$
   
In the general case, i.e. for $f\notin \Fo(\Hi_t)$, the two alternative definitions of the Feynman integral can yield different results, as we shall discuss in Sect. \ref{sez4}. It is interesting to remark that the integrability condition required by Def. \ref{intoscinf1} holds if the limit \eqref{infdim} is independent of the particular sequence of finite dimensional projection operators $\{P_n\}_n$, whereas the definition of Feynman map \eqref{feymap} relies upon the piecewise-linear approximations.


\section{Schr\"odinger equation with magnetic field}\label{sez3.0}

Let us consider the dynamics of a non-relativistic quantum particle moving in a magnetic field $\bB=\rot \ba$, where $\ba$ is a vector potential associated to $\bB$.  The quantum Hamiltonian operator for this system is given on smooth compactly supported vector $\psi \in C_0^\infty(\R^3) $ by 
\[
H\psi=\frac{1}{2}\left(-i\hbar\nabla -\lambda\ba(x)\right)^2\psi , \qquad \psi \in C_0^\infty(\R^3)
\]
(where, for notation simplicity, we set equal to $1$ the parameters mass $m$, velocity of light $c$, and elementary charge $e$). The parameter $\lambda \in \R$ stands for a coupling constant.  In the following we shall assume that   the components $a_j$, $j=1, 2,3$, of the vector potential $\ba$ are real valued functions and belong  to the space $\Fo_c(\R^3)$ of Fourier transforms of complex Borel measures $\mu_{j}$ on $\R^3$ with compact support, i.e. they are functions of the form
\begin{equation}\label{conditionaJ}
a_j(x)=\int_{\R^3} e^{iyx}d\mu _{j}(y), \qquad x\in \R^3, \quad  j=1,2,3,
\end{equation}
with $\mu_j$ of compact support.

\begin{Remark}
\label{RX}
For later use we point out that $a_j \in \Fo_c(\R^3)$, $j=1,2,3$ implies that $a_j$ has an  analytic continuation to a function on $\C^3$, denoted by the same symbol.
\end{Remark}

Under this assumption, it is possible to prove (see, e.g. \cite{PaOs,Sim,Cyc87,Lein81,BrHuLe}) that $H$ defined on $C_0^\infty(\R^3)$ is positive, symmetric and, $\ba $ being bounded, has  a unique self-adjoint extension $H:D(H)\subset L^2(\R^3)\to L^2(\R^3)$, with domain 
\begin{equation}\label{DomH}D(H)=\left\{\psi \in L^2(\R^3) \colon \int_{\R^3}|y|^4|\hat \psi (y) |^2dk<+\infty\right\},\end{equation} where $\hat \psi \in  L^2(\R^3)$ denotes the Fourier transform of the vector $\psi \in L^2(\R^3)$ and $|y|$ is the norm of the vector $y\in \R^3$.

By Stone's theorem,  $H$ generates a one-parameter group $U(t)=e^{-\frac{i}{\hbar}Ht}$, $t \in \R$, of unitary operators on $L^2(\R^3)$, solving 
the Schr\"odinger equation in the following sense:
\begin{equation}
\label{SchMFt}
i\hbar \partial _t U(t) \psi_0=HU(t)\psi_0, \qquad \psi_0\in D(H)\end{equation}
(where the derivation on the left is a strong one in $L^2(\R^3)$).

We shall assume, without loss of generality, that the vector potential $\ba$ satisfies the Coulomb gauge, namely that $\div \ba=0$. 
In this case the Hamiltonian operator  $H$ can be written as $H=H_0+W$, where $H_0$ is the free Hamiltonian, namely the operator  $H_0=-\frac{\hbar^2}{2}\Delta$ on $D(H_0)=D(H)$ (defined by \eqref{DomH}) and $W=\lambda A+\lambda ^2 B$, where $A= i\hbar \ba \cdot \nabla$ and $B $ is the multiplication operator associated to the function $\frac{1}{2}|\ba^2|$ (both well defined on $D(H_0)$).\\
 
Under the assumption that the initial datum $\psi_0\in L^2(\R^2)$ has a compactly supported Fourier transform, i.e. $\psi_0 \in \Fo_c(\R^3) \cap L^2(\R^3)$, the Dyson series expansion (in powers of the coupling constant $\lambda$ for the vector  $U(t)\psi_0$ has a finite radius of convergence as the following theorem shows. For related results by purely analytic methods see \cite{PaOs}).
\begin{teorema}\label{teoDyson}
Let us assume that the Fourier transform $\hat \psi _0$ of the vector $\psi_0\in L^2(\R^d)$  has a compact support included in the ball  $B_\rho\equiv \{x\subset \R^3\colon |x|<\rho\}$   and that  the component $a_j$, $j=1,\dots,3$, of the vector potentials $\ba$ are of the form \eqref{conditionaJ}, with $\mu_j$ , $j=1,\dots,3$, bounded Borel measures with support contained in the  ball $B_R$, for some $R\in \R^+$. 
Then the expansion in powers of the coupling constant $\lambda$ for the vector $U(t) \psi_0$, namely $U(t) \psi_0=\sum_m\lambda ^m \phi_m(t)$, with 
$$\phi_m (t)= \sum_{(n,k)\in \N^2 \colon 2n-k=m, \, k\leq n}\left(-\frac{i}{\hbar}\right)^n\phi_{n,k}$$
and $\phi_{n,k}$ given by \eqref{phinkprecise} below, converges in $L^2(\R^3)$ for $|\lambda|<\lambda ^*$, with 
\begin{equation}
\label{lambdastar}\lambda^*=\left(\frac{2\alpha ^2 t}{\hbar}\left(2r ^2 t\hbar+1\right)\right)^{-1/2}, 
\end{equation} 
where $\alpha =\sup_{x\in \R^3}|\ba(x)|$ and $r=\max\{\rho, R\}$.

\begin{Remark}
Correspondingly as in Remark \ref{RX}, we point out for later use that the assumption on $\psi_0$ in Theorem \ref{teoDyson} implies that $\psi_0$ has a unique extension to an analytic function on $\C^3$, denoted by the same symbol.
\end{Remark}

\end{teorema}
We start by proving the following lemma.
\begin{lemma}\label{lemmastimaAB}Under the assumptions of Theorem \ref{teoDyson}, the following holds:
\begin{equation}\|U_0(t_1)O_1U_0(t_2)O_2\cdots U_0(t_n)O_nU_0(t_{n+1})\psi_0 \|\leq \hbar ^k\alpha ^k \left(\frac{\alpha ^2}{2}\right)^{n-k}\prod_{j=0}^{k-1}(\rho +2R(n-k)+jR)\|\psi_0\|,
\end{equation}
where $U_0(t)=e^{-\frac{i}{\hbar }H_0t}$,  $O_j\in \{A,B\}$, $j=1,\dots,n$,   $k=\#\{j\colon O_j=A\}$, with $A=i\hbar \ba\cdot\nabla$,  $B=\frac{1}{2}|\ba|^2$ and $\alpha =\sup_{x\in \R^3}|\ba(x)|$.
\end{lemma}
\begin{proof}
Let $\phi\in L^2(\R^3)$  be a function whose Fourier transform $\hat \phi$ has support in a ball $B_\rho$ centered at the origin with radius $\rho \in \R^+$. 
Then the following holds.
\begin{itemize}
\item For any $t\in \R$, the vector $U_0(t) \phi\in L^2(\R^3)$  has a Fourier transform with support contained in $B_\rho$. Indeed $\widehat{U_0(t) \phi}$ is simply given by $$\widehat{U_0(t) \phi}(y)=e^{-\frac{i}{\hbar}|y|^2 t}\hat \phi (y), \qquad y\in \R^3.$$
\item The vector $B\phi\in L^2(\R^3)$   has a Fourier transform with support contained in $B_{\rho+2R}$. Indeed, under the assumptions on the components $a_j$ of the vector field $\ba$,  the function $x\mapsto |\ba |^2$ is Fourier transform of a Borel measure $\mu_{\ba^2}$ on $\R^3$ with support contained in the ball $B_{2R}$, with $\mu_{\ba^2}=\sum_{j=1}^3\mu_j * \mu_j$, where the symbol $\mu*\nu$ stands for the convolution of the measures $\mu $ and $\nu$.  It is simple to verify that if the supports of the measures $\mu_j$ is contained in $B_R$, then the support of the convolution $\mu_j* \mu_j$ is contained in $B_{2R}$. Correspondingly, the Fourier transform of $B\phi$ is given by 
$$\widehat{B\phi}(y)=\frac{1}{2}\int_{\R^3} \hat \phi (y-y')d\mu_{\ba^2}(y') , \quad y\in \R^3,$$ 
and its support is contained in $B_{\rho+2R}$.
\item The norm of the vector $B\phi$ is bounded by 
$$\|B\phi\|\leq \frac{1}{2}\|\ba^2 \|_\infty\|\phi\|,$$
where $\|\ba^2 \|_\infty=\sup_{x\in \R^3}|\ba (x)|^2$,  which is finite by the assumptions on the components $a_j$, $j=1,\dots,3$.
\item The vector $A\phi\in L^2(\R^3)$ (with $A$ as in Lemma \ref{lemmastimaAB}) has a Fourier transform with support contained in $B_{\rho+R}$, given by 
$$\widehat{A\phi}(y)=-\hbar\sum_{j=1}^3\int_{\R^3}(y_j-y'_j)\hat \phi (y-y')d\mu_j(y'), \qquad k\in \R^3.$$
Moreover, the norm of $A\phi$ satisfies the following bound
$$\|A\phi\|\leq \hbar \sqrt{\|\ba^2 \|_\infty}\, \rho\, \|\phi\|.$$
\end{itemize}
Now it is straightforward to verify that, if $\#\{j\colon O_j=A\}=k$: 
\begin{gather*}
\|U_0(t_1)O_1U_0(t_2)O_2\cdots U_0(t_n)O_nU_0(t_{n+1})\psi_0 \| \leq\\  \leq \|U_0(t_1)AU_0(t_2)\cdots U_0(t_k)A U_0(t_{k+1})B\cdots U_0(t_n) BU_0(t_{n+1})\psi_0 \| \leq \\
\leq \hbar ^k\alpha ^k \left(\frac{\alpha ^2}{2}\right)^{n-k}\prod_{j=0}^{k-1}(\rho +2R(n-k)+jR)\|\psi_0\|.
\end{gather*}
\end{proof}

\DDD
By the classical Dyson expansion for the vector $U(t)\psi_0$ \cite{PaOs}, we have 
\begin{align*}
U(t)\psi_0& =\sum_{n =0} ^\infty \left(-\frac{i}{\hbar}\right)^n\int_{\Delta_n(t)}e^{-\frac{i}{\hbar}H_0(t-s_n)}We^{-\frac{i}{\hbar}H_0(s_n-s_{n-1})}\cdots We^{-\frac{i}{\hbar}H_0(s_2-s_1)}We^{-\frac{i}{\hbar}H_0s_1}\psi_0 ds_1\dots ds_n \\
&=\sum_{n =0} ^\infty \left(-\frac{i}{\hbar}\right)^n\int_{\Delta_n(t)}U_0(t-s_n)(\lambda A +\lambda^2B)U_0(s_n-s_{n-1})\cdots(\lambda A +\lambda^2B)U_0(s_1)\psi_0 ds_1\dots ds_n,
\end{align*}
where $\Delta_n(t)\subset \R^n$ is the $n$-dimensional simplex defined as $\Delta_n(t)=\{(s_1, \dots, s_n ) \in \R^n \colon 0\leq s_1\leq \dots\leq s_n\leq t\}$ and $W=\lambda A +\lambda^2B$.
The dependence on the coupling constant $\lambda$ can be made explicit as
\begin{equation}\label{De1}
U(t)\psi_0 =\sum_{n =0} ^\infty \left(-\frac{i}{\hbar}\right)^n\sum_{k=0}^{n}\lambda ^k (\lambda ^2)^{n-k} \phi_{n,k},
\end{equation}
where the term $\phi_{n,k}\in L^2(\R^3)$ is a sum of $ \binom{n}{k}$ terms of the form
$$
\int_{\Delta_n(t)}U_0(t-s_n)O_1 U_0(s_n-s_{n-1})O_2\cdots O_n U_0(s_1)\psi_0 ds_1\dots ds_n,
$$
where we recall that $O_j=A,B$, $j=1,\dots,n$ and $\#\{j\colon O_j=A\}=k$. More precisely:
\begin{equation}
\label{phinkprecise}
\phi_{n,k}=\sum_{E\in C^n_k}\int_{\Delta_n(t)}U_0(t-s_n)O^E_1 U_0(s_n-s_{n-1})O^E_2\cdots O^E_n U_0(s_1)\psi_0 ds_1\dots ds_n
\end{equation}
where the sum is taken over the set $C^n_k$ of all possible subsets $E\subset\{1,\dots, n\}$ with $k$ elements and the map $ O^E:\{1,\dots, n\}\to \{A, B\}$ is defined as $O^E_i:=A$ if $i\in E$ and $O^E_i:=B$ if $i\notin E$.

 By Lemma \ref{lemmastimaAB} we have 
\begin{equation}
\label{stima-phi-nk}
\|\phi_{n,k}\|\leq  \binom{n}{k}\frac{t^n}{n!}\hbar ^k\alpha ^k \left(\frac{\alpha ^2}{2}\right)^{n-k}\prod_{j=0}^{k-1}(\rho +2R(n-k)+jR)\|\psi_0\|.
\end{equation}
In particular, by setting $r:=\max\{\rho, R\}$, we obtain:
\begin{equation}
\label{stima-phi-nk-2}
\|\phi_{n,k}\|\leq  \frac{t^n}{(n-k)!}\hbar ^k\alpha ^{2n-k} \left(\frac{1}{2}\right)^{n-k}r^k{\binom{2n-k}{k}}\|\psi_0\|.
\end{equation}
Now, the sum appearing in \eqref{De1} can be written as:
\begin{align}
U(t)\psi_0 &=\sum_{m=0} ^\infty \lambda ^m\sum_{(n,k)\in \N^2 \colon 2n-k=m, \, k\leq n}\left(-\frac{i}{\hbar}\right)^n\phi_{n,k}\nonumber\\
&=\sum_{m}\lambda ^m \phi_m. \label{seriesDyson-m}
\end{align}
where 
\begin{align}
\label{phimprec-eo}
\begin{array}{l}
\phi_m=\left(-\frac{i}{\hbar}\right)^{m/2}\sum_{h=0}^{m/2}\left(-\frac{i}{\hbar}\right)^h\phi_{h+\frac{m}{2},2h},\qquad m \hbox{   even};\\
\ \\
\phi_m=\left(-\frac{i}{\hbar}\right)^{(m+1)/2}\sum_{h=0}^{(m-1)/2}\left(-\frac{i}{\hbar}\right)^h\phi_{h+\frac{m+1}{2},2h+1},\qquad m \hbox{   odd}. 
\end{array}
\end{align}

By estimate \eqref{stima-phi-nk-2} we have for $m$ even, namely $m=2M$, $M \in \N$:

\begin{align*}
\|\phi_{2M}\| &\leq \left(\frac{\alpha ^2 t}{2\hbar}\right)^M\sum_{h=0}^M \left(2r ^2 t\hbar\right)^h \frac{1}{(M-h)!}{\binom{2M}{2h}}\|\psi_0\|\\
&\leq \left(\frac{\alpha ^2 t}{2\hbar}\right)^M \left(2r ^2 t\hbar+1\right)^M\max_{h\in \{0,\dots,M\}} \frac{h!}{M!} {\binom{2M}{2h}}\|\psi_0\|\\
&\leq\left(\frac{2\alpha ^2 t}{\hbar}\left(2r ^2 t\hbar+1\right)\right)^M\|\psi_0\|;
\end{align*}
analogously, for $m$ odd, namely $m=2M+1$, $M \in \N$ we get
\begin{align*}
\|\phi_{2M+1}\| &\leq  2rt\alpha\hbar \left(\frac{2\alpha ^2 t}{\hbar}\left(2r ^2 t\hbar+1\right)\right)^M\|\psi_0\|.
\end{align*}
Hence, the series \eqref{seriesDyson-m} converges in $L^2(\R^d)$ for
$|\lambda|<\left(\frac{2\alpha ^2 t}{\hbar}\left(2r ^2 t\hbar+1\right)\right)^{-1/2}.
$

\finedim

\begin{Remark}\label{rem-anal} Since the Hamiltonian operator $H$ is self-adjoint and positive, as pointed out at the beginning of this section, it generates an analytic semigroup. Hence, for $z\in \C$ belonging to the closure $\bar D$ of the open sector $D\subset \C$ of the complex plane defined as
$$D=\{z\in \C\colon Re(z)> 0\},$$it is possible to define the operator $V(z)=e^{-zH}$ yielding for $z=\frac{i}{\hbar }t$ and $t\in \R$ the Schr\"odinger group and for $z\in \R^+$ the heat semigroup.  In both cases, under the assumptions of Theorem \ref{teoDyson} the perturbative Dyson expansion for the vector $V(z)\psi_0$ has a positive radius of convergence (depending on $|z|$). Indeed, if $z\in D$, the Dyson expansion can be written as
\begin{align}
e^{-zH}\psi_0& =\sum_{n =0} ^\infty (-z)^n \int_{\Delta_n}e^{-zH_0(1-s_n)}We^{-zH_0(s_n-s_{n-1})}\cdots We^{-zH_0(s_2-s_1)}We^{-zH_0s_1}\psi_0 ds_1\dots ds_n \label{Dysonanal1} 
\end{align}
with  $\Delta_n\equiv \Delta_n(1)=\{(s_1, \dots, s_n ) \in \R^n \colon 0\leq s_1\leq \dots\leq s_n\leq 1\}$ and $W=\lambda A+\lambda ^2B$. By collecting in the sum \eqref{Dysonanal1} all the  terms associated to the same power of the coupling constant $\lambda$, we get
\begin{equation}\label{dyson-z}e^{-zH}\psi_0 =\sum_{m}\lambda ^m \phi_m(z)\end{equation}
where 
$ \phi_m(z)=\sum_{(n,k)\in \N^2\colon 2n-k =m, \, k\leq n}(-z)^n\phi_{n,k}(z)$
with 
\begin{equation}
\label{phinkprecise-z}
\phi_{n,k}(z)=\sum_{E\in C^n_k}\int_{\Delta_n}e^{-zH_0(t-s_n)}O^E_1 e^{-zH_0(s_n-s_{n-1})}O^E_2\cdots O^E_n e^{-zH_0s_1}\psi_0 ds_1\dots ds_n,
\end{equation}
where, analogously to Eq. \eqref{phinkprecise}, the sum is taken over the set $C^n_k$ of all possible subsets $E\subset\{1,\dots, n\}$ with $k$ elements and the map $ O^E:\{1,\dots, n\}\to \{A, B\}$ is defined as $O^E_i:=A$ if $i\in E$ and $O^E_i:=B$ if $i\notin E$.

By repeating the arguments in the proof of Theorem \ref{teoDyson}, it is now easy to verify that the expansion \eqref{dyson-z} converges in $L^2(\R^3) $ for $|\lambda |< \lambda ^*(z)$, with 
\begin{equation}
\label{lambda-z}
\lambda ^*(z)=\left(2\alpha ^2 |z|\left(2r ^2\hbar^2 |z|+1\right)\right)^{-1/2}, \qquad z \in D.
\end{equation}
\end{Remark}

\begin{Remark}\label{RemDysonV}
The results of Theorem \ref{teoDyson} and Remark \ref{rem-anal} can  easily be extended to the case where a (bounded) scalar potential $V$ is added to the Hamiltonian. Indeed, let us consider the following
\begin{equation}\label{H-a-V}
H\psi(x)=\frac{1}{2}\left(-i\hbar\nabla -\lambda\ba(x)\right)^2\psi (x)+\lambda V(x)\psi (x) , \qquad \psi \in C_0^\infty(\R^3)
\end{equation}
where $V\in \Fo_c(\R^3)$, i.e. $V:\R^3\to \R$ a function of the form:
\begin{equation}
\label{hypV}
V(x)=\int e^{ixy}d\mu_V(y), \qquad x\in \R^3,
\end{equation}
with $\mu_V$ complex Borel measure with support contained in the ball $B_R$. Under the assuptions of Theorem \ref{teoDyson}, Lemma \ref{lemmastimaAB} still holds. In particular, by setting  $A=i\hbar \ba\cdot\nabla+V$, $B=\frac{1}{2}|\ba|^2$, $\alpha =|\ba|_\infty$ and $\tilde\alpha =2\max\{\hbar|\ba|_\infty,|V|_\infty \}$, we get:
\begin{equation}\label{Dys-V}\|U_0(t_1)O_1U_0(t_2)O_2\cdots U_0(t_n)O_nU_0(t_{n+1})\psi_0 \|\leq \tilde\alpha ^k \left(\frac{\alpha ^2}{2}\right)^{n-k}\prod_{j=0}^{k-1}(\rho +2R(n-k)+jR)\|\psi_0\|,
\end{equation}
where $U_0(t)=e^{-\frac{i}{\hbar }H_0t}$,  $O_j\in \{A,B\}$, $j=1,\dots,n$,   $k=\#\{j\colon O_j=A\}$. By using \eqref{Dys-V} it is now possible to repeat the proof of Theorem \ref{teoDyson}, obtaining the convergence in $L^2(\R^3)$ of the perturbative Dyson expansion for the vector $e^{-\frac{i}{\hbar}Ht}\psi_0$ for $\lambda<\tilde \lambda$, where 
\begin{equation}
\label{lambdatilde}\tilde\lambda=\left(\frac{2\tilde\alpha ^2 t}{\hbar}\left(\frac{2r ^2 t}{\hbar}+1\right)\right)^{-1/2}. 
\end{equation} 
\end{Remark}


\section{Feynman path integral for magnetic field}
\label{sez3}

The present section is devoted to the construction of the Feynman path integral representation of the solution to the time dependent Schr\"odinger equation \eqref{SchMFt} in terms of Feynman maps on the Cameron-Martin space $\Hi_t$ defined in Sect. \ref{sez2}, i.e.
\begin{equation}\label{Sch1.0}
\psi (t,x)=\int _{\Hi_t}e^{\frac{i}{2\hbar}\int_0^t\|\dot\gamma(s)\|^2ds-\frac{i}{\hbar}\int_0^t  \lambda \ba(\gamma(s)+x)\cdot \dot \gamma (s)ds}\psi_0(\gamma(t)+x)d\gamma.
\end{equation}

\begin{Remark}
\label{R7}
Formula \eqref{Sch1.0} differs from \eqref{Sch1} for the sign in front of the term  $\int_0^t  \lambda \ba(\gamma(s)+x)\cdot \dot \gamma (s)ds$. This is due to the fact that in the heuristic Feynman formula \eqref{Sch1}  the paths $\gamma$ are pointed at the final time (i.e. $\gamma(t)=x$), while in Eq. \eqref{Sch1.0} the path $\gamma\in \Hi_t$ satisfy the condition $\gamma (t)=0$.
\end{Remark}

First of all, it is interesting to point out that the existing techniques of infinite dimensional oscillatory integration based on Parseval-type equality (see Theorem \ref{teo-pars-1}) do not work in the case where the classical action functional contains the term $\int_0^t  \ba(\gamma(s))\cdot \dot \gamma (s)ds$(as the term in the exponent of \eqref{Sch1.0}). 

In fact the function $f:\Hi_t\to \C$ defined on vectors $\gamma$ belonging to the Cameron-Martin space $\Hi_t$ as 
\begin{equation}
\label{f}
f(\gamma):=\int_0^t \ba(\gamma(s))\cdot \dot \gamma (s)ds, \qquad \gamma \in \Hi_t,
\end{equation} cannot in general belong to the  Banach algebra $\Fo(\Hi_t)$,
 even under  rather strong assumption on the vector potential $\ba$, unless in the trivial case where $\ba$ would be a conservative vector field (hence the associated magnetic field $\rot \ba$ would vanish identically!). In this case indeed  it is simple to prove that one has $f\in \Fo(\Hi_t)$ and $f(\gamma):=\int_0^t \ba(\gamma(s))\cdot \dot \gamma (s)ds=U(\gamma(t))-U(\gamma(0)) $). This particular case has already been studied in \cite{AlCaMa18}. However, in the physically more interesting case where $\rot \ba \not\equiv 0$, even if any of the three components $a_i$, $i=1,2,3$, of the vector potential $\ba$ belongs to $\Fo(\R^3)$, it is not possible to prove that $f\in \Fo(\Hi_t)$. In fact, the oscillatory integration of function $f$ involves most of the problems arising in  stochastic integration theory \cite{KarSh}. Indeed,  since any $\gamma \in \Hi_t$ is a bounded variation function, it is easy to show that the function $f$ is the pointwise limit of sequence of cylinder functions of the form
 $$f_n(\gamma):=\sum_{j=0}^{n-1} \ba(\gamma(t_j))\cdot (\gamma(t_{j+1})-\gamma(t_j)), \qquad \gamma \in \Hi_t$$
 or, equivalently, of
 $$g_n(\gamma):=\sum_{j=0}^{n-1}  \ba(\gamma(t_{j+1}))\cdot (\gamma(t_{j+1})-\gamma(t_j)),\qquad \gamma \in \Hi_t,$$
 where $t_j\equiv jt/n$.  Furthermore, if  $a_i\in \Fo(\Hi_t)$, then the cylinder functions  $\{f_n\}_n$ and $\{g_n\}_n$  are Fresnel integrable since they are both finite  linear combinations of functions of the form \eqref{fresELT}.  Indeed , if  for any $i=1,..,3$  $a_i=\hat \mu_i$, with $\mu_i$ bounded complex Borel measures on $\R^3$, then
$$f_n (\gamma):=\sum_{j=0}^{n-1} \sum_{\alpha =1}^3\int_{\R^3}e^{i\langle \gamma ,v_{t_j}k\rangle}\langle \gamma,(v_{t_{j+1}}-v_{t_j})\hat e_\alpha\rangle d\mu_\alpha (k)$$
and 
$$g_n( \gamma):=\sum_{j=0}^{n-1} \sum_{\alpha =1}^3\int_{\R^3}e^{i\langle \gamma ,v_{t_{j+1}}k\rangle}\langle \gamma,(v_{t_{j+1}}-v_{t_j})\hat e_\alpha\rangle d\mu_\alpha (k),$$
where $\hat e_\alpha$, $\alpha =1,\dots,3$ are the vectors of the canonical basis of $\R^3$, while for $s\in [0,t]$ the function $v_s:[0,t]\to \R$ is defined by
\begin{equation}\label{defvs}v_s(r)=\chi_{[0,s]}(r) r+\chi_{(s,t]}(r) s, \qquad r\in [0,t].\end{equation}
By direct computation, one has $\widetilde{\int_{\Hi_t}}e^{\frac{i}{2\hbar}\langle \gamma , \gamma \rangle }f_n(\gamma)d\gamma=0$ for all $n\in \N$, while 
$$\widetilde{\int_{\Hi_t}}e^{\frac{i}{2\hbar}\langle \gamma , \gamma \rangle }g_n(\gamma)d\gamma=-\hbar \sum_{j=0}^{n-1} \sum_{\alpha =1}^3\int_{\R^3}e^{-\frac{i\hbar}{2}t_{j+1}|k|^2}(t_{j+1}-t_j)k_\alpha d\mu_\alpha (k)$$
the latter converging, for $n\to \infty$ to $-\hbar \sum_{\alpha =1}^3\int_{\R^3}\int_0^te^{-\frac{i\hbar}{2}s|k|^2}k_\alpha d\mu_\alpha (k)ds$.

Since  the Parseval type equality \eqref{Parseval} cannot be directly applied, we have to implement a different technique, based on analyticity assumptions, in order to show that the limit in definition \eqref{feymap} exists, i.e. that the Feynman map of the function $f$ given by \eqref{f} is well defined.

In the following we shall denote with $C_t:=C([0,t];\R^3)$ the Banach space of continuous paths $\omega:[0,t]\to \R^3$, endowed with the $\sup$-norm $|\;|$. Let $\bP$ be the Wiener measure on the Borel $\sigma$-algebra $\Ba(C_t)$ of $C_t$.
Since for $\gamma\in \Hi_t$ we have $|\gamma|\leq \sqrt{t} \cdot \|\gamma\|$, the Cameron-Martin Hilbert space $\Hi_t$ is densely embedded in $C_t$. Denoted with $C_t^*$ the topological dual of $C_t$, we have the following chain of dense inclusions:
\begin{equation}
\label{Ct}
C_t^*\subset \Hi_t\subset C_t.
\end{equation}
With an abuse of notation we shall denote $\langle \eta , \omega \rangle$ the dual pairing between  two elements $\eta\in C_t^*$ and $\omega \in C_t$. Let $\mu$ be the finitely additive standard Gaussian measure defined as 
$$
\mu(\cC_{P_n,D})=\int_{D}\frac{e^{-\frac{\|x\|^2}{2}}}{(2\pi )^{n/2}}dx,$$ on the cylinder sets $\cC_{P_n,D}\subset \Hi_t$ of the form $$
\cC_{P_n,D}:=\{\gamma \in \Hi_t\colon P_n\gamma\in D\},
$$
for some finite dimensional projection operator $P_n:\Hi_t \to \Hi_t$ and some Borel set $D\subset \Hi_t$. The measure $\mu$ does not extend to a $\sigma$-additive  measure on the generated $\sigma$-algebra, see e.g.\cite{Kuo}.  Defining the cylinder sets in $C_t$ by
$$\tilde \cC_{\eta_1,\dots,\eta_n;E}:=\{\omega \in C\colon (\langle \eta_1,\omega\rangle,\dots,\langle \eta_n,\omega\rangle)\in E\},$$
for some $n\in \N$, $\eta_1, \dots, \eta_n\in C_t^*$ and $E$ a Borel set of $\R^3$, we have that the intersection $\tilde \cC_{\eta_1,\dots,\eta_n;E}\cap\Hi_t$ is a cylinder set in $\Hi_t$. According to the fundamental results by L. Gross \cite{Gro1,Gro2}, the finite additive measure $\tilde \mu$ defined on the cylinder sets of $C_t$ by
$$\tilde \mu(\tilde \cC_{\eta_1,\dots,\eta_n;E}):=\mu(\tilde \cC_{\eta_1,\dots,\eta_n;E}\cap \Hi_t)$$
extends to a $\sigma$-additive Borel measure on $C_t$ that coincides with the standard  Wiener measure $\bP$, in such a way that  for any $\gamma\in \Hi_t$ such that $\gamma$  is an element of $C_t^*$ the following holds
$$
\int e^{i\langle\gamma, \omega \rangle}d\bP(\omega)=e^{-\frac{1}{2}\|\gamma\|^2}.
$$

Thanks to the results above it is possible to define, for any $\eta \in C_t^*$, a centered Gaussian random variable $n_\eta$ on $(C_t,\Ba(C_t),\bP)$ given by $n_\eta (\omega ):=\langle\gamma, \omega \rangle$, $\omega \in C_t$, $\gamma \in C_t^*$. In particular, for $\eta ,\gamma\in C_t^*$, the following holds
\begin{equation}\label{neta}
\bE[n_\eta n_\gamma]=\int_0^t\dot \eta(s)\cdot \dot \gamma (s)ds=\langle \eta, \gamma \rangle,
\end{equation}
the pairing on the r.h.s. coinciding with the scalar product in $\Hi_t$. This shows that the map $n:C_t^*\to L^2(C_t,\bP)$ can be extended, by the density of $C_t^*$ in $\Hi_t$, to an unitary operator $n:\Hi_t\to L^2(C_t,\bP)$. 
In particular, given a projector operator $P_n:\Hi_t\to \Hi_t$ of the form $P_n(\gamma)=\sum_{j=1}^n \langle e_n,\gamma\rangle e_n$, where $\{e_1, \dots, e_n \} $ orthonormal vectors in $\Hi_t$, it is possible to define the random variable $\tilde P_n:C_t\to \Hi_t$ as 
\begin{equation}
\label{tildePn}
\tilde P_n(\omega)=\sum_{i=1}^nn_{e_i}(\omega) e_i,
\end{equation}
$n_{e_i} \in L^2(C_t, \mathbb{P})$.


In the following we shall show how Feynman maps (defined by Eq. \eqref{feymap}) of all powers of  the function $f$ defined in \eqref{f} can be computed in terms of Wiener integrals. For analogous results see \cite{AlMa2}.
Let us consider now in $\Hi_t$ the sequence of projection operators $\{P_n\}$  onto the subspaces of piecewise linear paths, i.e. for $\gamma \in \Hi_t$ the vector $P _n(\gamma)$ is defined by the right hand side of \eqref{gamman}. Let $\{\tilde P_n\}$  the corresponding sequence of random variables $\tilde P_n:C_t\to \Hi_t$ given by 
\begin{equation}\label{omegandef}\tilde P_n(\omega) (s)=\sum_{k=1}^n\chi_{[t_{k-1},t_k]}(s)\left(\omega (t_{k-1})+\frac{\omega(t_{k})-\omega(t_{k-1})}{t_{k}-t_{k-1}}(s-t_{k-1})\right),\qquad s\in [0,t],
\end{equation}
$\omega \in C_t$, with $t_k=kt/n$, $k=1,\dots,n$ as above.
Let $\ba:\R^3\to \R^3$  be a vector field fulfilling the assumptions of  Theorem \ref{teoDyson}. Since any component $a_j:\R^3\to \R$, $j =1,\dots,3$,   can be written as the Fourier transform of a complex measure $\mu_j$ with compact support according to formula \eqref{conditionaJ}, the map $\ba$ can be extended to an holomorphic function on $\C^3$ with components given by 
\begin{equation}
\label{az}
\ba_j(z)=\int _{\R^3} e^{ikz}d\mu _j(k), \qquad z\in \C^3,
\end{equation}
the integral on the r.h.s. of \eqref{az} being well defined and finite since 
$$\int _{\R^3} | e^{ikz}|d|\mu _j|(k)\leq \int _{\R^3}\Pi_{l=1}^3  e^{|k_l ||z_l|}d|\mu _j|(k)\leq e^{R\sum_{l=1}^3|z_l|},$$
where $R$ denotes the radius of the sphere containing the supports of the measures $\mu_j$.
 In particular, for $x\in \R^3$ and $z\in \C$ the components of the vector $\ba (zx)$ are given by $\ba_j(zx)=\int _{\R^3} e^{izkx}d\mu _j(k)$. The following lemma shows the convergence of a particular  sequence of random variables defined on the Wiener space.

\begin{lemma}
\label{LemmaAppendix}
Let $\ba$ be a three dimensional vector field fulfilling the assumptions of Theorem \ref{teoDyson}. 
Let $\{f_n\}$ be  the sequence of random variables $f_n:C_t\to\C$ defined by \[
f_n(\omega)=\int_0^t \ba \left (  \sqrt{ i\hbar} \omega_n(s) \right )\cdot \dot{\omega}_n(s) ds, 
\]
where $\omega_n(s)\equiv P_n(\omega )(s)$ and $P_n(\omega) $ is defined by the right hand side of \eqref{omegandef}.
Then for any $p\in \N$, $1\leq p\leq \infty$, $f_n$ 
converges, as $n \to \infty$, in $L^p(C_t,\mathbb{P})$ to the random variable $f$ defined as the Stratonovich stochastic integral 
\[
f(\omega)=\int_0^t \ba ( \sqrt{ i\hbar} \omega(s))\circ d\omega(s).
\]
\end{lemma}
\begin{proof}
We will consider for notational simplicity the $1-$dimensional case. The proof in all dimensions, in particular in the 3-dimensional case is analogous. Let us first remark that by Remark \ref{RX} on the analyticity of the extension of $\ba(\cdot)$ from $\R$ to $\C$, the integral on the right hand side  of $f_n$ is well defined. Further, by \eqref{az} the random variables $f_n$ are given by:

\begin{align*}
f_n({\omega(s)})&=\sum_{j=0}^{n-1} \int_0^{\frac{t}{n}}\int_{\R} e^{i \sqrt{ i\hbar} k {\omega(s_j)}}e^{i \sqrt{ i\hbar}k\frac{({\omega(s_j)}j-{\omega(s_j)})s}{t/n}} \cdot \frac{({\omega(s_j)}j-{\omega(s_j)})}{t/n} d\mu(k) ds =\\
&=\sum_{j=0}^{n-1} \int_{\R} e^{i \sqrt{ i\hbar}k{\omega(s_j)}}\left ( e^{i \sqrt{ i\hbar}k({\omega(s_j)}j-{\omega(s_j)})}-1 \right )\cdot \frac{1}{i \sqrt{ i\hbar}k} d\mu(k).
\end{align*}

By setting $\Delta_j:={\omega(s_j)}j-{\omega(s_j)}$ and by a Taylor expansion (to second order with remainder), the last line becomes
$$\sum_{j=0}^{n-1} \int_{\R} e^{i \sqrt{ i\hbar}k{\omega(s_j)}} \left ( \Delta_j+\frac{1}{2} i \sqrt{ i\hbar}k\Delta_j^2+\frac{1}{2} (i \sqrt{ i\hbar}k)^2\Delta_j^3\int_0^1(1-u)^2e^{i \sqrt{ i\hbar}k\Delta_ju}du\right ) d\mu(k).$$

Hence the function $f_n$ can be written as the sum of three contributions, namely $f_n=g_n+h_n+r_n$, where
\begin{align}
g_n(\omega)=& \sum_{j=0}^{n-1} \int_{\R} e^{i \sqrt{ i\hbar}k{\omega(s_j)}} ({\omega(s_j)}j-{\omega(s_j)}) d\mu(k)= \nonumber \\ =& \sum_{j=0}^{n-1} a( \sqrt{ i\hbar}{\omega(s_j)})  ({\omega(s_j)}j-{\omega(s_j)});\\
h_n(\omega)=&\frac{1}{2}\sum_{j=0}^{n-1} \int_{\R}  i \sqrt{ i\hbar}ke^{i \sqrt{ i\hbar}k{\omega(s_j)}}({\omega(s_j)}j-{\omega(s_j)}) ^2d\mu(k) \nonumber \\
=& \sum_{j=0}^{n-1} \frac{1}{2}\cdot a'( \sqrt{ i\hbar} {\omega(s_j)})({\omega(s_j)}j-{\omega(s_j)})^2;  \\
r_n(\omega)=& \sum_{j=0}^{n-1} \frac{1}{2}\int_0^1 \int_{\R} \left ( i \sqrt{ i\hbar}\right)^2k^2e^{i \sqrt{ i\hbar}k({\omega(s_j)} + ({\omega(s_j)}j-{\omega(s_j)})u)}({\omega(s_j)}j-{\omega(s_j)})^3(1-u)^2d\mu(k)du,
\end{align}
($a'$ standing for derivative of $a$).
By computation based on BDG inequalities and Gaussian integration we obtain
\begin{align*}
g_n & \xrightarrow{L^p(\Omega,\mathbb{P})} \int_0^t a( \sqrt{ i\hbar} \omega(s)) d\omega(s); \\
h_n & \xrightarrow{L^p(\Omega,\mathbb{P})} \frac{1}{2} \int_0^t a'( \sqrt{ i\hbar} \omega(s))ds; \\
r_n & \xrightarrow{L^p(\Omega,\mathbb{P})} 0,
\end{align*}
eventually obtaining:
\[
f_n  \xrightarrow{L^p(\Omega,\mathbb{P})} a( \sqrt{ i\hbar} \omega(s)) d\omega(s)+\frac{1}{2} \int_0^t a'( \sqrt{ i\hbar} \omega(s))ds=\int_0^t a( \sqrt{ i\hbar} \omega(s))  \circ d\omega(s).
\]
For further details  see Appendix A.
\end{proof}

\begin{teorema}
\label{Teorema3}
Let the vector field $\ba $ and the function $\psi_0\in L^2(\R^3)$ satisfy the assumptions of Theorem \ref{teoDyson}. Then the Feynman map of the function $g:\Hi_t\to \C$ given by $g(\gamma):=\psi_0(\gamma (t)+x)\int_0^t\ba (\gamma (s)+x)\cdot \dot \gamma (s)ds$ for any $x \in \R^3$, is well defined and equal to the following Wiener integral 
\begin{equation}
\label{int-f-1}
I_F(g)=\int_{C_t}\left( \sqrt{ i\hbar}\int_0^t  \ba \left (  \sqrt{ i\hbar}\omega(s)+x \right )\circ d\omega(s) \right)  \psi_0( \sqrt{ i\hbar}\omega(t)+x)d\mathbb{P}(\omega).
\end{equation}
Moreover for any $m\geq 0$,  the Feynman map of the function $g_m^x:\Hi_t\to \C$  defined as $$g_m^x(\gamma):=\psi_0(\gamma (t)+x)\left(\int_0^t\ba (\gamma (s)+x)\cdot \dot \gamma (s)ds\right)^m$$ is given by
\begin{equation}
\label{int-f-n}
I_F(g_m^x)=\int_{C_t}\left( \sqrt{ i\hbar}\int_0^t  \ba \left (  \sqrt{ i\hbar}\omega(s)+x \right )\circ d\omega(s) \right)^m  \psi_0 \left ( \sqrt{ i\hbar}\omega(t)+x \right )d\mathbb{P}(\omega),
\end{equation}
$\mathbb{P}$ being Wiener measure on $(C_t,\mathcal{B}(C_t))$.
\end{teorema}

\begin{proof}
For Fixed $n\in \N$ and $m\geq 1$, let us consider the finite dimensional oscillatory integral
\begin{gather*}
\left(  \int _{P_n\Hi}^o e ^{i\frac{\|P_n \gamma \|^2}{2\hbar} }d(P _n\gamma )\right )^{-1}  \int _{P_n\Hi}^o e ^{i\frac{\|P _n \gamma \|^2}{2\hbar} }g_m(P_n\gamma)d (P_n\gamma )=\\
=  \int _{\R^{3n}}^o \left(\sum_{j=1}^n\frac{(x_{j}-x_{j-1})}{t/n}\cdot\int_{t_{j-1}}^{t_j}\ba\left(x_{j-1}+\frac{(x_{j}-x_{j-1})}{t/n}(s-t_{j-1} )+x\right) ds\right)^m\\ 
\psi_0(x_n+x)e^{\frac{i}{2\hbar t/n}\sum_{j=1}^n(x_{j}-x_{j-1})^2} \frac{dx_1\cdots dx_n}{(2\pi i \hbar t/n)^{3n/2}}=\\
=  \int_{\R^{3n}}^o \left(\sum_{j=1}^n\xi_j\cdot\int_{t_{j-1}}^{t_j}\ba\left(x+ \left ( \sum_{k=1}^{j-1}\xi_k \right )\frac{t}{n}+\xi_j(s-t_{j -1}) \right) ds\right)^m \psi_0\left (x+\left (\sum_{j=1}^n\xi_j\right ) \frac{t}{n}\right )\\
e^{\frac{it/n}{2\hbar}\sum_{j=1}^n\xi_{j}^2} \frac{d\xi_1\cdots d\xi_n}{(2\pi i \hbar (t/n)^{-1})^{3n/2}}.
\end{gather*}

By the stated assumption on $\ba$ and $\psi_0$ the latter is equal to
\begin{equation}
\label{int-n-0}
\begin{gathered}
\int_{\R^{n}}^o \left( \sum_{j=1}^n\sum_{\alpha=1}^3\xi_{j,\alpha}\cdot\int_{t_{j-1}}^{t_j}\int_{\R^d} \exp\left [ ik\cdot \left(x+ \left (\sum_{l=1}^{j-1}\xi_l \right )\frac{t}{n}+\xi_j(s-t_{j-1} )\right) \right ] d\mu_\alpha(k)ds\right)^m \\ \int_{\R^3}e^{ih\cdot \left ( x+\sum_{j=1}^n\xi_jt/n \right )}d\mu_0(h) e^{\frac{it/n}{2\hbar}\sum_{j=1}^n\xi_{j}^2} \frac{d\xi_1\cdots d\xi_n}{(2\pi i \hbar (t/n)^{-1})^{3n/2}}.
\end{gathered}
\end{equation}
where $\hat\mu_0=\psi_0$, i.e for any Borel set $I\subset \R^3$ $\mu_0(I)=\frac{1}{2\pi}\int_I\psi _0(x)dx$.
Let us consider the open sector  $D_{\pi/2}=\{z\in \C\colon z=|z|e^{i\theta}, \theta \in (0,\pi/2)\}$ of the complex plane and function $F:\bar D_{\pi/2}\to \C$ defined by 
\begin{equation*}
\begin{gathered}
F(z)=\int _{\R^{3n}} \left( \sum_{j=1}^n\sum_{\alpha=1}^3z\xi_{j,\alpha}\cdot\int_{t_{j-1}}^{t_j}\int_{\R^3}\exp \left [ ik\cdot \left( x+ \left ( \sum_{l=1}^{j-1}z\xi_l \right ) \frac{t}{n}+z\xi_j(s-t_{j-1} )\right) \right ] d\mu_\alpha(k)ds \right)^m \\ \int_{\R^3}e^{ih\cdot (x+z\sum_{j=1}^n\xi_jt/n)}d\mu_0(h) e^{\frac{it/n}{2\hbar}\sum_{j=1}^nz^2\xi_{j}^2} \frac{d\xi_1 \cdots d\xi_n}{(2\pi i \hbar z^{-1}(t/n)^{-1})^{3n/2}},
\end{gathered}
\end{equation*}
for $z\in \R$, $z>0$, by the classical change of variable formula, $F(z)$ is a constant function equal to the finite dimensional oscillatory integral \eqref{int-n-0}.  Further $F $ is analytic on $D_{\pi/2}$, as one can prove by applying Fubini and Morera's theorems. Indeed, for $z\in D_{\pi/2}$, the integral defining $F(z)$ is absolutely convergent since:
\begin{gather}
\int _{\R^{3n}} \left(\sum_{j=1}^n\sum_{\alpha=1}^3|z\xi_{j,\alpha}|\int_{t_{j-1}}^{t_j}\int_{\R^3}\left | \exp \left [ik\cdot \left(x+ \left( \sum_{l=1}^{j-1}z\xi_l \right ) \frac{t}{n}+z\xi_j(s-t_{j-1})\right) \right ]\right | d|\mu_\alpha|(k)ds\right)^m \nonumber \\  
\int_{\R^3}\left | e^{ih\cdot \left (x+z\sum_{j=1}^n\xi_j\frac{t}{n} \right )}\right |d|\mu_0|(h) e^{\frac{it/n}{2\hbar}\sum_{j=1}^nz^2\xi_{j}^2} |\frac{d\xi_1 \cdots d\xi_n}{(2\pi  \hbar |z|^{-1}(t/n)^{-1})^{3n/2}} \leq \nonumber \\
\leq 
\int_{\R^{3n}} \left(\sum_{j=1}^n\sum_{\alpha=1}^3|z| |\xi_{j,\alpha}|\int_{0}^{t/n}\int_{\R^3}\exp\left[-|z|\sin \theta k\cdot \left(\sum_{l=1}^{j-1}\xi_l \frac{t}{n}+\xi_j s )\right) \right]d|\mu_\alpha|(k)ds\right)^m \nonumber
\\
\int_{\R^3}\left | e^{-|z| \sin \theta h\cdot \sum_{j=1}^n\xi_j \frac{t}{n}}\right | d|\mu_0|(h) e^{-\frac{|z|^2\sin (2\theta)t/n}{2\hbar}\sum_{j=1}^nz^2\xi_{j}^2} \frac{d\xi_1 \cdots d\xi_n}{(2\pi \hbar |z|^{-1}(t/n)^{-1})^{3n/2}}= \nonumber \\
=
\int_{\R^{3n}} \left(\sum_{j=1}^n\sum_{\alpha=1}^3|\xi_{j,\alpha}|\int_{0}^{t/n}\int_{\R^3}\exp\left[-\sin \theta k\cdot \left(\sum_{l=1}^{j-1}\xi_l \frac{t}{n}+\xi_j s )\right) \right]|d|\mu_\alpha|(k)ds\right)^m \nonumber \\  
\int_{\R^3}\left | e^{-\sin \theta h\cdot \sum_{j=1}^n\xi_j\frac{t}{n}}\right | d|\mu_0|(h) e^{-\frac{\sin (2\theta)t/n}{2\hbar}\sum_{j=1}^n\xi_{j}^2} \frac{d\xi_1 \cdots d\xi_n}{(2\pi  \hbar (t/n)^{-1})^{3n/2}} \leq \nonumber \\
\leq 
\Bigg (\int _{\R^{3n}} \left(\sum_{j=1}^n\sum_{\alpha=1}^3|\xi_{j,\alpha}|\int_{0}^{t/n}\int_{\R^3}\exp\left[-\sin \theta k\cdot \left(\sum_{l=1}^{j-1}\xi_l \frac{t}{n}+\xi_j s )\right) \right]|d|\mu_\alpha|(k)ds\right)^{2m} \nonumber
\\ 
e^{-\frac{\sin (2\theta)t/n}{2\hbar}\sum_{j=1}^n\xi_{j}^2} \frac{d\xi_1\cdots d\xi_n}{(2\pi  \hbar (t/n)^{-1})^{3n/2}}\Bigg)^{1/2}
\nonumber\\ 
\left(\int_{\R^{3n}} \left ( \int_{\R^3}\left | e^{-\sin \theta h\cdot \sum_{j=1}^n\xi_j \frac{t}{n}}\right | d|\mu_0|(h) \right )^2e^{-\frac{\sin (2\theta)t/n}{2\hbar}\sum_{j=1}^n\xi_{j}^2} \frac{d\xi_1\cdots d\xi_n}{(2\pi \hbar (t/n)^{-1})^{3n/2}}\right)^{1/2}. \label{longF1}
\end{gather}
In the second step above we have got rid of the term $|z|$ in the integral because of classical (finite dimensional) change of variables formula.\\
For notational simplicity, in the following we shall describe in detail the one dimensional case but  similar arguments work also in three dimension.
The second factor in the product of integrals above is bounded by
\begin{gather*}
\int_{\R^{n}} \left ( \int_{\R}\left | e^{-\sin \theta h\cdot \sum_{j=1}^n\xi_j \frac{t}{n}}\right | d|\mu_0|(h) \right )^2e^{-\frac{\sin (2\theta)t/n}{2\hbar}\sum_{j=1}^n\xi_{j}^2} \frac{d\xi_1 \cdots d\xi_n}{(2\pi \hbar (t/n)^{-1})^{n/2}}=\\
= (\sin (2\theta))^{-n/2}\int_\R\int_\R e^{\frac{\hbar t \sin ^2 \theta (h_1+h_2)^2}{2\sin (2\theta)}}d|\mu_0|(h_1) d|\mu_0|(h_2) \leq \\
\leq (\sin (2\theta))^{-n/2} e^{\frac{2\hbar t \sin ^2 \theta R^2}{\sin (2\theta)}}|\mu_0|^2,
\end{gather*}
where $R\in \R^+$ is such that the support of $\mu_0$ is contained in $[-R,R]$.\\
Concerning the first factor on the right hand side of \eqref{longF1} we have the following upper bound, again written for simplicity of notations for the $1$-dimensional case  
\begin{gather*}
\int_{\R^{n}} \left(\sum_{j=1}^n|\xi_{j}|\int_{0}^{t/n}\int_{\R}\exp\left[ -\sin \theta k\cdot \left(\sum_{l=1}^{j-1}\xi_l \frac{t}{n}+\xi_j s \right) \right] d|\mu|(k)ds\right)^{2m} 
\\ 
e^{-\frac{\sin (2\theta)t/n}{2\hbar}\sum_{j=1}^n\xi_{j}^2} \frac{d\xi_1 \cdots d\xi_n}{(2\pi \hbar (t/n)^{-1})^{n/2}}= 
\\
=\sum_{j_1,\dots ,j_{2m}=1}^n\int_{\R^{n}}|\xi_{j_1}\cdots \xi_{j_{2m}}|\int_{0}^{t/n}\dots \int_{0}^{t/n}\int_{\R} \dots  \int_{\R} \exp\left [ -\sin \theta k_1 \cdot \left(\sum_{l_1=1}^{j_1-1}\xi_{l_1}\frac{t}{n}+\xi_{j_1 }s_1 \right) \right ] \cdots
\\ 
\cdots \exp \left[ -\sin \theta k_{2m} \cdot \left( \sum_{l_{2m}=1}^{j_{2m}-1}\xi_{l_{2m}}\frac{t}{n} +\xi_{j_{2m}} s_{2ms} \right) \right ] ds_1 \cdots ds_{2m} d|\mu|(k_1) \cdots  d|\mu|(k_{2m})
\\
e^{-\frac{\sin (2\theta)t/n}{2\hbar}\sum_{j=1}^n\xi_{j}^2} \frac{d\xi_1\cdots d\xi_n}{(2\pi  \hbar (t/n)^{-1})^{n/2}} \le \\
\leq \sum_{j_1,\dots ,j_{2m}=1}^n\left(\int _{\R^{n}}|\xi_{j_1}\cdots \xi_{j_{2m}}|^2e^{-\frac{\sin (2\theta)t/n}{2\hbar}\sum_{j=1}^n\xi_{j}^2} \frac{d\xi_1\cdots d\xi_n}{(2\pi  \hbar (t/n)^{-1})^{n/2}}\right)^{1/2}
\\
\Bigg(\int _{\R^{n}} \Bigg ( \int_{0}^{t/n}\dots\int_{0}^{t/n}\int_{\R}\dots \int_{\R} \exp\left[-\sin \theta k_1\cdot \left(\sum_{l_1=1}^{j_1-1}\xi_{l_1}\frac{t}{n}+\xi_{j_1 }s_1 )\right) \right]\cdots 
\\
\cdots \exp\left[ -\sin \theta k_{2m}\cdot \left(\sum_{l_{2m}=1}^{j_{2m}-1}\xi_{l_{2m}}\frac{t}{n}+\xi_{j_{2m}} s_{2ms} )\right) \right] ds_1 \cdots ds_{2m}d|\mu|(k_1)\cdots d|\mu|(k_{2m}) \Bigg )^2
\\
e^{-\frac{\sin (2\theta)t/n}{2\hbar}\sum_{j=1}^n\xi_{j}^2} \frac{d\xi_1 \cdots d\xi_n}{(2\pi  \hbar (t/n)^{-1})^{n/2}}\Bigg)^{-1/2}.
\end{gather*}
The first factor in the sum above is  finite since it is  equal to the moment of a Gaussian measure, i.e.
\[
\int_{\R^{n}}|\xi_{j_1}\cdots \xi_{j_{2m}}|^2e^{-\frac{\sin (2\theta)t/n}{2\hbar}\sum_{j=1}^n\xi_{j}^2} \frac{d\xi_1 \cdots d\xi_n}{(2\pi  \hbar (t/n)^{-1})^{n/2}}=(\sin(2\theta)^{-n/2})\cdot \E \left [ \left |X_{j_1} \cdots X_{j_{2m}} \right |^2 \right ],
\]
where $X_j$, $j=1,\dots ,n$ are i.i.d centered Gaussian random variables with covariance $\sigma =\hbar(\sin(2\theta)t/n)^{-1}$.
Analogously the second factor is an absolutely convergent integral since it is of the form
\begin{gather*}
\int_{0}^{t/n}\dots\int_{0}^{t/n}\int_{\R}\dots \int_{\R} 
    \E \left [    e^{\sum_{j=1}^na_j(k_1,\dots ,k_{4m},s_1,\dots ,s_{4m})X_j} \right ]ds_1\cdots ds_{4m}d|\mu|(k_1)\cdots d|\mu|(k_{4m})=
\\
=\int_{0}^{t/n}\dots\int_{0}^{t/n}\int_{\R}\dots \int_{\R} e^{\frac{1}{2\sin(2\theta)t/n}\sum_{j=1}^n\left ( a_j(k_1,\dots ,k_{4m},s_1,\dots ,s_{4m}) \right )^2}ds_1\cdots ds_{4m}d|\mu|(k_1)\cdots d|\mu|(k_{4m}),
\end{gather*}
where $a_j$ are linear functions of the variables $k_1,\dots ,k_{4m},s_1, \dots ,s_{4m}$ and  the last integral is finite since $\mu$ is  by assumption compactly supported. Hence we can conclude that $F$ is analytic on $D_{\pi/2}$. Further, by a classical change of variables formula, it is simple to see that $F$ is constant on any ray of the form $r_\theta:=\left \{z\in \C: z=|z|e^{i\theta}, |z|\in \R^+ \right \}$ with $\theta \in [0,\pi/2]$, hence by analyticity it is constant on $\bar D_{\pi/2}$, giving for any $n\in \N$
\begin{gather}
\left(  \int_{P _n\Hi}^o e ^{i\frac{\|P _n \gamma \|^2}{2\hbar} }d (P _n\gamma )\right )^{-1}  \int_{P_n\Hi}^o e ^{i\frac{\|P_n \gamma \|^2}{2\hbar} }g_m(P _n \gamma )d (P_n\gamma )= F(e^{i\pi /4}) = \nonumber
\\
= \int _{\R^{3n}} \left(\sum_{j=1}^n \sqrt{ i\hbar}\,\frac{(x_{j}-x_{j-1})}{t/n}\cdot\int_{t_{j-1}}^{t_j}\ba\left( \sqrt{ i\hbar}\,x_{j-1}+ \sqrt{ i\hbar}\,\frac{(x_{j}-x_{j-1})}{t/n}(s-t_{j-1} )+x\right) ds\right)^m \nonumber
\\
\psi_0\left (  \sqrt{ i\hbar}\, x_n+x \right )e^{-\frac{1}{2 t/n}\sum_{j=1}^n(x_{j}-x_{j-1})^2} \frac{dx_1 \cdots dx_n}{(2\pi \hbar t/n)^{3n/2}}= \nonumber
\\
= {\mathbb E}\left[  \psi_0\left (  \sqrt{ i\hbar}\, \omega_n(t)+x \right )\left( \sqrt{ i\hbar}\,\int_0^t\ba\left( \sqrt{ i\hbar}\omega_n(s)+x\right)\cdot  \dot \omega_n(s)ds\right)^m \right], \label{longF2}
\end{gather}
where $\omega_n$ stands for the piecewise linear approximation of Brownian motion defined above, namely:
$$
\omega_n(s)=\sum_{k=1}^n\chi_{[t_{k-1},t_k]}(s)\left(\omega (t_{k-1})+\frac{\omega(t_{k})-\omega(t_{k-1})}{t_{k}-t_{k-1}}(s-t_{k-1})\right),
$$
with $s\in [0,t]$, $t_k=k/n$. Thanks to the result of Lemma \ref{LemmaAppendix}, the right side of \eqref{longF2} converges  for $n \to \infty$ to 
$$ {\mathbb E}\left[  \psi_0( \sqrt{ i\hbar}\, \omega(t)+x)\left( \sqrt{ i\hbar}\,\int_0^t\ba\left( \sqrt{ i\hbar}\omega(s)+x\right)\circ  d \omega(s)\right)^m \right].
$$
\end{proof}

The last step is the proof that the integral provides a representation of the solution to the Schr\"odinger equation by the Dyson series expansion. 

\begin{teorema}
\label{Teorema4}
Under the assumption of Theorem \ref{teoDyson}, the solution of the Schr\"odinger equation with magnetic field
\[
i \hbar \partial_t \psi(t)=H \psi(t,x), \quad \psi(0,x)=\psi_0(x), \quad H=\frac{1}{2}(-i\hbar \nabla - \lambda \ba(x))^2
\]
can be expressed by  the perturbative Dyson series expansion as
\[
e^{-\frac{i}{\hbar}Ht}\psi_0=\sum_{m=0}^{\infty}\lambda ^m \psi_m(t),
\]
where the vector  $\psi_m$ can be expressed by a Feynman map of the form
\begin{align}
\psi_m(t,x)=& \frac{1}{m!}\left (-\frac{i}{\hbar} \right )^m \widetilde{\int_{\mathcal{H}_t}}\left(\int_0^t  \ba (\gamma(s)+x)\cdot \dot \gamma(s)ds\right )^m e^{\frac{i}{2\hbar}\int_0^t\|	\dot\gamma(s)\|^2ds}\psi_0(\gamma(t)+x) d\gamma\nonumber\\
=& \frac{1}{m!}\left (-\frac{i}{\hbar} \right )^m \mathbb{E} \left [ \left( \sqrt{ i\hbar}\int_0^t  \ba \left ( \sqrt{ i\hbar}\omega(s)+x \right )\circ d\omega(s) \right)^m  \psi_0\left ( \sqrt{ i\hbar}\omega(t)+x \right ) \right ].\label{inf-n-fin}
\end{align}
The expansion is convergent in $L^2(\R^3)$ for $\lambda \in \C$, with $|\lambda |<\lambda ^*$, $\lambda ^*$ given by \eqref{lambdastar}.
The integral under the expectation is to be understood as a Stratonovich stochastic integral.
\end{teorema}

\begin{proof} 
By Theorem \ref{teoDyson} for $|\lambda|<\lambda ^*$ the vector $\psi(t)=e^{-\frac{i}{\hbar }Ht}\psi_0$ in $L^2(\R^3)$ is given by the convergent power series expansion \eqref{seriesDyson-m}. Hence, we are left to prove that for any $m\in \N$ the term $\phi_m$ in \eqref{phimprec-eo} is equal to $\psi_m$ as given in \eqref{inf-n-fin}.

By Remark \ref{rem-anal}, the Hamiltonian operator $H$ generates an analytic semigroup $e^{-zH}:L^2(\R^3)\to L^2(\R^3)$, where $z\in \C$, $Re(z)\geq 0$, with a  convergent Dyson expansion  of the form $e^{-zH}\psi_0=\sum _m \lambda ^m \phi_m (z)$ with a radius of convergence $\lambda^*(z)$ depending on $|z|$ (see Eq. \eqref{lambda-z}).
In particular, for $z\in \R^+$, namely $z=\frac{t}{\hbar}$, the family of operators $T(t)=e^{-\frac{t}{\hbar}H}$, $t\in \R^+$, yields the heat semigroup generated by $H$ (described in Section \ref{sez3.0}). In this case, by Feynman-Kac-It\=o formula \cite{Sim} the vector $e^{-\frac{t}{\hbar}H}\psi_0$ is given by the Wiener integral
\begin{equation}
\label{heat1}
e^{-\frac{t}{\hbar}H}\psi_0(x)=\bE \left [\psi_0(\sqrt \hbar \omega (t)+x)e^{-\frac{i\lambda}{\hbar}\int_0^t\ba (\sqrt\hbar \omega (s)+x)\circ d\omega (s)} \right ].
\end{equation}
For any $\phi\in L^2(\R^3)$ the inner product $\langle \phi, e^{-zH}\psi_0\rangle$ is an analytic function of $z\in D$, $D=\{z\in \C, Re(z)\geq 0\}$,  continuous in the closure $\bar D$ of $D$ and admitting the expansions
$$\langle \phi, e^{-zH}\psi_0\rangle=\sum_{m=0}^{\infty}\lambda ^m \langle \phi, \phi_m(z)\rangle. $$
By formula \eqref{phinkprecise-z} each term $\langle \phi, \phi_m(z)\rangle$ is an analytic function of $z\in D$, continuous in the closure $\bar D$ and for $z=t/\hbar$, $t\in \R^+$, by formula \eqref{heat1} it is equal to
\begin{equation}\label{heat2}
\langle \phi, \phi_m(t/\hbar)\rangle=\left(-\frac{i}{\hbar}\right)^m\int_{\R^3}\bar \phi(x)\bE\left[\psi_0(\sqrt \hbar \omega (t)+x)\left(\sqrt\hbar\int_0^t\ba (\sqrt\hbar \omega (s)+x)\circ d\omega (s)\right)^m\right]dx.
\end{equation}
By replacing $t$ with $t\xi$, with  $\xi\in \R^+$, the expression above assumes the following form:
\begin{equation*}
\langle \phi, \phi_m(t\xi/\hbar)\rangle=\left(-\frac{i}{\hbar}\right)^m\int_{\R^3}\bar \phi(x)\bE\left[\psi_0(\sqrt \hbar \omega (t\xi)+x)\left(\sqrt\hbar \int_0^{t\xi}\ba (\sqrt\hbar \omega (s)+x)\circ d\omega (s)\right)^m\right]dx,
\end{equation*}
and thus
\begin{equation}
\label{heat3}
\scalebox{0.965}{$\displaystyle \langle \phi, \phi_m (t\xi/ \hbar )\rangle =\left(-\frac{i}{\hbar}\right)^m\int_{\R^3}\bar \phi(x)\bE\left[\psi_0(\sqrt{ \hbar\xi} \omega (t)+x)\left(\sqrt{ \xi \hbar}\int_0^{t}\ba (\sqrt{\hbar \xi}\omega (s)+x)\circ d\omega (s)\right)^m\right]dx,$}
\end{equation}

and since by the discussion above, both the right hand side and the left hand side of \eqref{heat3} are analytic for $\xi \in D$ and continuous for $\xi\in \bar D$, by setting $\xi\equiv i$ we obtain the required equality, namely:
$$\langle \phi, \phi_m(it/\hbar)\rangle=\left(-\frac{i}{\hbar}\right)^m\int_{\R^3}\bar \phi(x)\bE\left[\psi_0(\sqrt{ i\hbar} \omega (t)+x)\left(\sqrt{ i\hbar} \int_0^{t}\ba (\sqrt{i\hbar }\omega (s)+x)\circ d\omega (s)\right)^m\right]dx.$$
\end{proof}

\begin{Remark}\label{RemDysonV-1}

Theorem \ref{Teorema4} can be generalized to the case where a scalar potential $V$  is added to the right hand side of \eqref{SchMF}.
 Indeed, let  us  consider an Hamiltonian operator of the form \eqref{H-a-V}, with $V\in \Fo_c(\R^3)$. By Remark \ref{RemDysonV}, 
under the assumptions of Theorem  \ref{Teorema4} the vector $e^{-\frac{i}{\hbar}Ht}\psi_0 $ admits for $\lambda <\tilde \lambda$ ($\tilde\lambda $ defined as in \eqref{lambdatilde}) a convergent perturbative expansion: \[
e^{-\frac{i}{\hbar}Ht}\psi_0=\sum_{m=0}^{\infty}\lambda ^m \psi_m(t),
\]
where the generic vector  $\psi_m$ can be expressed by a Feynman map of the form
\begin{equation*}
\scalebox{0.935}{
$
\displaystyle
\psi_m(t,x)= \frac{1}{m!}\left (-\frac{i}{\hbar} \right )^m \widetilde{\int_{\mathcal{H}_t}}\left(\int_0^t  \ba (\gamma(s)+x)\cdot \dot \gamma(s)ds +\int_0^tV(\gamma (s) +x)ds     \right)^m e^{\frac{i}{2\hbar}\int_0^t\|	\dot\gamma(s)\|^2ds}\psi_0(\gamma(t)+x) d\gamma,
$}
\end{equation*}
which can be expressed in terms of the Wiener integral
$$\frac{1}{m!}\left (-\frac{i}{\hbar} \right )^m \mathbb{E} \left [ \left( \sqrt{ i\hbar}\int_0^t  \ba \left ( \sqrt{ i\hbar}\omega(s)+x \right )\circ d\omega(s)+\int_0^tV(\sqrt{i\hbar}\omega (s) +x)ds \right)^m  \psi_0\left ( \sqrt{ i\hbar}\omega(t)+x \right ) \right ].$$

 \end{Remark}

\begin{Remark}\label{Rem-10}
All results in Theorems \ref{teoDyson}, \ref{Teorema3}, \ref{Teorema4} have been formulated and proved for underlying $3$-dimensional space, but corresponding results and proofs hold for all space dimensions.
\end{Remark}


\section{Independence of the approximation procedure and renormalization term}
\label{sez4}

In the previous section we provided a convergent constructive expansion for the Feynman path integral representation for the solution of the Schr\"odinger equation with magnetic field. This was made by using a particular class of finite dimensional approximations, namely the ones related to piecewise linear path (see Eq. \eqref{gamman}).  This last section is devoted to the question, whether the independence of the construction of the Feynman path integral representation is independent of the chosen type of approximation. In particular, in the case of a constant magnetic field, we  show that the definition of  the Feynman path integral \eqref{Sch1} in terms of infinite dimensional oscillatory integral (in the sense of Def. \ref{intoscinf1}), i.e. requiring the independence of the limit of the sequence of finite dimensional approximations, requires the introduction of a natural renormalization term.
This result is a further development of a similar one obtained in \cite{AlBrMF}, the latter being only valid in the Coulomb gauge $\div \ba=0$. On the contrary, our main results (Theorem \ref{const3D} and Corollary \ref{Cor1}) provide a gauge-independent construction of the renormalization term as well as of the Feynman path integral, yelding a rigorous construction of the solution of the Schr\"odinger equation with vector potential $\ba$.

Let $\ba:\R^3\to \R^3$ be a linear vector potential corresponding to a constant magnetic field $\bB=\rot\ba$. More precisely, we assume that the vector field $\ba$ is given by 
\begin{equation}\label{vector-a}
\ba(x_1,x_2,x_3)=(\alpha^1_1x_1+\alpha^1_2x_2+\alpha^1_3x_3,\alpha^2_1x_1+\alpha^2_2x_2+\alpha^2_3x_3,\alpha^3_1x_1+\alpha^3_2x_2+\alpha^3_3x_3), \quad (x_1,x_2,x_3)\in \R^3, 
\end{equation}
where $\alpha ^i_j\in \R$, $i,j=1,\dots,3$ are real constants.  
We are going to study the Fresnel integrability (in the sense of Def. \ref{intoscinf1}) of the function $f:\Hi_t\to\C$, defined on the Cameron-Martin space $\Hi_t$ as
$$f(\gamma):=e^{-\frac{i}{\hbar}\int_0^t\ba(\gamma(s))\cdot \dot \gamma (s)ds}, \qquad \gamma \in \Hi_t. $$

For  any sequence $\{P _n\}_n$ of projectors
onto n-dimensional subspaces of $\Hi_t$, such that $P _n\leq P_{n+1}$ and $P_n \to \mathbb{I}$ strongly as $ n \to \infty$, we have to study the limit of the sequence of finite dimensional oscillatory integrals
$$\lim_{ n \to \infty} (2\pi i \hbar )^{-n/2}  \int_{P _n\Hi_t}^oe ^{i\frac{\|P _n \gamma\|^2}{2\hbar} }f(P _n \gamma )d (P _n\gamma ).
$$ 
As we shall see, the limit above cannot be independent of the sequence $\{P_n\}$. In fact it is necessary to renormalize the term $f(P_n\gamma) \equiv e^{-\frac{i}{\hbar}g(P_n\gamma)}$ by replacing the exponent $g(P_n\gamma) =\int_0^t\ba(P_n\gamma(s))\cdot \dot P_n\gamma (s)ds$ by $g(P_n\gamma)-r_n$, where $r_n$ is a suitable constant depending on the projector $P_n$ as well as on the magnetic field $§\bB$.

First of all, let us consider the linear operator $G:\Hi_t\to \Hi_t$ defined by 
\begin{equation}
\label{op-G}
G(\gamma)(s):=\int_0^s \ba (\gamma(r))dr, \qquad \gamma \in \Hi_t, \;s\in [0,t],
\end{equation}
in such a way that the function $f:\Hi_t\to \C$ can be written as $f(\gamma )=e^{-\frac{i}{\hbar}\langle G(\gamma), \gamma\rangle}$, i.e. the function $g:\Hi_t\to \C$, with $g(\gamma)= \langle G(\gamma), \gamma\rangle$ can be represented as the quadratic form associated to $G$. Note that $G$ is bounded in $\Hi_t$ due to our assumptions on $\ba$. The following lemma provides some properties of $G$.
\begin{lemma}\label{lemmaG}
The operator $G:\Hi_t\to\Hi_t$ is Hilbert-Schmidt .
The eigenvalues of the positive symmetric operator $G^\dag G$, are given by $\lambda _{m,j}= \frac{4 a_jt^2}{\pi^2(1+2m)^2}$, with $j=1,2,3,$ $m\in \N$ and $a_j\in \R^+$ eigenvalues of the matrix \eqref{matrixA} below.
\end{lemma}
\begin{proof}Let us consider the positive symmetric operator $L\equiv G^\dag G:\Hi_t\to\Hi_t$ ($\dag$ standing for the adjoint), whose matrix elements are given by 
\begin{eqnarray}
\langle \eta, L\gamma\rangle &=&
\langle G\eta, G \gamma\rangle \nonumber
\\
&=&\int_0^t\eta(s)A \gamma(s)^Tdt, \qquad \eta, \gamma \in \Hi_t,  \label{IntA}
\end{eqnarray}
where   $A$ is the $3\times 3$ symmetric matrix with real elements given by 
\begin{equation}\label{matrixA}
A_{ij}=\sum_{l=1}^3\alpha ^l_i\alpha^l_j, \qquad i,j=1,...,3.
\end{equation}
Hence, for $\gamma \in \Hi_t$ the vector $L(\gamma)\in \Hi_t$ is given by
$$
L(\gamma )(s)^T=-\int_0^s\int_t^rA\gamma (\tau)^Td\tau dr,
$$
$T$ stands for transpose. $L$ is a compact operator in $\Hi_t$ and has a discrete spectrum. 
Indeed, by introducing in $\mathbb{R}^3$  an orthonormal basis $\{\hat u_1,\hat u_2,\hat u_3\} $ of eigenvectors of the  symmetric matrix $A$, with corresponding eigenvalues $a_1, a_2, a_3 \in \mathbb{R}^+$, the eigenvectors $\{ \gamma_{m} \}$ of $L$ can be represented as linear combination of $\hat u_1$, $\hat u_2$ and $\hat u_3$ as  $\gamma_{m}=\eta_{m,1}\hat u_1+\eta_{m,2}\hat u_2+\eta_{m,3}\hat u_3$, with $\eta_{m,j}:[0,t]\to \mathbb{R}$. Recalling the form of the scalar product in $\Hi_t$, for the expression \eqref{IntA} we get that the  components $\{\eta_{m,j}\}$ of the eigenvectors (with eigenvalues $\lambda_{m,j}$) are the solutions of 
\begin{equation*}
\begin{cases}
\lambda_{m,j}\ddot \eta_{m,j}+a_j\eta_{m,j}=0\\
\dot \eta_{m,j}(t)=0\\
\eta_{m,j}(0)=0
\end{cases}
\qquad j=1,2,3,
\end{equation*}
with
\begin{equation}
\label{eigenvalues}
\lambda_{m,j}= \frac{4 a_jt^2}{\pi^2(1+2m)^2}, \qquad m\in \N, \quad j=1,2,3. 
\end{equation}

\end{proof}

\begin{Remark}\label{rem-notrace}
From \eqref{eigenvalues} we see easily that the operator $G$ is not trace class in $\Hi_t$.
\end{Remark}

\begin{lemma}\label{lemma-anal-linear}
Let $\ba$ be the linear vector potential \eqref{vector-a} and let $\psi_0 \in L^2(\R^3)$ be such that its Fourier transform $\hat \psi_0$ has compact support.  Let $\{e_j\}_j$ be an orthonormal basis of the Cameron-Martin space $\Hi_t$ and let $P_n$ be the projection operator onto the span of the first $n$ vectors. Let $g_{\exp}^x:\Hi_t \to \C$ the function defined as
\[
g_{\exp}^x(\gamma)= \psi_0(\gamma(t)+x)\exp \left(-\frac{i}{\hbar}\int_0^t\ba (\gamma (s) +x) \cdot \dot \gamma (s)ds \right ), \qquad \gamma \in \Hi_t,
\]
and let $\bar{a}\in \R^+$ be the  constant defined as $\bar{a}=\displaystyle{\max_{j=1,2,3}} \{ a_j \}$, where $a_j$, $j=1,2,3$ are the eigenvalues of the (positive semidefinite) matrix $A$ defined in \eqref{matrixA}.
Then, for fixed $n$ and for
\begin{equation}\label{tstar}
t < t^*:=\frac{\pi}{4\sqrt{\bar{a}}}, 
\end{equation}
 the finite dimensional oscillatory integral $$(2\pi i\hbar )^{-n/2}\int_{P_n\Hi_t}^oe^{\frac{i}{2\hbar}\|P_n\gamma\|^2} g^x_{\exp}(P_n\gamma) dP_n\gamma$$ is equal to the    Wiener integral:
\begin{equation}\label{fey-Wie-findim}
\int_{P_n\Hi_t}^oe^{\frac{i}{2\hbar}\|\gamma\|^2} g^x_{\exp}(\gamma) d\gamma =\E \left [\psi_0(\sqrt{ i\hbar}{\omega}_n(t)+x)e^{-\frac{i}{\hbar}\sqrt{ i\hbar}\int_0^t\ba (\sqrt{ i\hbar}{\omega}_n(s) +x) \cdot \dot{{\omega}}_n(s) ds} \right ],
\end{equation}
where ${\omega}_n:=\tilde{P}_n(\omega)$ is defined by \eqref{tildePn}.
\end{lemma}
\begin{proof}
By definition we have, for fixed $n$, by setting $\gamma_n\equiv P_n\gamma$:
\begin{align*}
 \left( 2\pi i \hbar \right )^{-n/2} \int _{P _n\Hi_t}^o e ^{i\frac{\|\gamma_n \|^2}{2\hbar} }g^x_{exp}(\gamma_n)d \gamma_n 
 &=\left( 2\pi i \hbar \right )^{-n/2} \int_{P_n\Hi}^o e^{\frac{i}{2\hbar} \|\gamma_n \|^2 }e^{-\frac{i}{\hbar}\int_0^t\ba(\gamma_n(s)+x)\dot\gamma_n ds}\psi_0(\gamma_n(t)+x)d \gamma _n\\
&=\left( 2\pi i \hbar \right )^{-n/2}\int_{P_n\Hi}^o e^{\frac{i}{2\hbar} \langle \gamma_n,\left(\mathbb{I}-2G\right)\gamma_n \rangle}e^{-\frac{i}{\hbar}\ba(x)\cdot \gamma_n(t)}\psi_0(\gamma_n(t)+x)d\gamma_n.
\end{align*}
Let us consider the function $F\to\R^+\to \C$ defined by 
\begin{equation}
\label{Fz}
F(z)=\left( 2\pi i \hbar  \right )^{-n/2}z^n\int_{P_n\Hi}^o e^{\frac{iz^2}{2\hbar} \langle \gamma_n,\left(\mathbb{I}-2G\right)\gamma_n \rangle}e^{-z\frac{i}{\hbar}\ba(x)\cdot \gamma_n(t)}\psi_0(z\gamma_n(t)+x)d\gamma_n.
\end{equation}

 By the classical change of variable formula, for $z \in \R^+$ the function $F$ is a constant  equal to the finite dimensional oscillatory integral above. In fact, if $t < t^*$, with $t^*$ given by \eqref{tstar},   $F$ can be extended to an analytic function defined on  the open sector  $D_{\pi/2}=\{z\in \C\colon z=|z|e^{i\theta}, \theta \in (0,\pi/2), |z|>0\}$ of the complex plane. Indeed, for any $\gamma \in \Hi_t$, if condition \eqref{tstar} is fulfilled, we have 
 
$$ \langle \gamma,\left(\mathbb{I}-2G\right)\gamma \rangle\geq \epsilon \|\gamma \|^2,$$ where $\epsilon >0$ is given by $\epsilon =1-\frac{2t\sqrt{\bar a}}{\pi}$. Indeed:

 \begin{align*}
\langle\gamma, \left( \mathbb{I}-2G\right ) \gamma \rangle &=\langle \gamma, \mathbb{I} \gamma \rangle - \langle \gamma, 2G\gamma \rangle\\
&= \langle\gamma,\gamma \rangle - \langle\gamma, 2|G|U \gamma \rangle.
\end{align*}
where, by the polar decomposition formula, $G=|G|U$, with $|G|=\sqrt{G^\dag G}$ and $U$ unitary operator.
Furthermore  
$$
|\langle\gamma, 2|G|U \gamma \rangle|\leq2 \|\gamma\|\|U\gamma\| \| |G|\|\leq 2\|\gamma \|^2 \sup_m \tilde\lambda _m,
$$
where $\| |G|\|$ denotes the operator norm of the positive operator$|G|$, while $\{\tilde\lambda_m\}_m $ are its eigenvalues, namely $\tilde\lambda_m=\sqrt{\lambda _{m,j}}$, with $\lambda_{m,j}$ given by \eqref{eigenvalues}. Hence, we get
$$|\langle\gamma, 2|G|U \gamma \rangle|\leq \frac{4t\sqrt{\bar a}}{\pi}\|\gamma \|^2, $$
hence, for $t<t^*$, we have, using the Fourier transform $\hat{\psi}_0$ of $\psi$,the following bound on the oscillatory integral in \eqref{Fz}:
\begin{gather*}
\int_{P_n\Hi}^o \left | e^{z^2 \frac{i}{2\hbar} \langle \gamma_n,\left(\mathbb{I}-2G\right)  \gamma_n \rangle} e^{-z\frac{i}{\hbar}\ba(x)\cdot \gamma_n(t)}\psi_0(z\gamma_n(t)+x)\right | d\gamma_n  \le \\
\le \int_{P_n\Hi}^o\int_{\R^3} e^{-\sin (2\theta)\frac{|z|^2}{2\hbar}\epsilon \|\gamma \|^2 +\sin \theta \frac{|z|}{\hbar}\ba (x)\cdot \gamma_n(t) -k|z|\gamma_n(t)}\frac{|\hat{\psi}_0(k)|}{(2\pi)^3}dk d\gamma_n <\infty,
\end{gather*}
where the convergence of the integral in the second line is assured by the conditions $\theta\in (0, \pi/2)$, $\epsilon =1-\frac{2t\sqrt{\bar a}}{\pi}>0$ and $\hat \psi_0$ compactly supported. 
Hence, by applying Fubini and Morera's theorems, it is simple to check that the function $F:\bar D_{\pi/2}\to \C$ is analytic on $D_{\pi/2}$ and continuous up to $\R^+$. Since by the classical change of variables formula the value of $F(z) $ does not depend on $|z|$, i.e. $F$ is constant along rays $\{ z\in D_{\pi/2}\colon z=|z|e^i\theta, |z|\in \R^+\}$, by analyticity $F$ is constant on $D_{\pi/2}$ and by continuity up to $\R^+$ we obtain in particular, $F(1)=F(\sqrt \hbar e^{i\pi/4})$, namely:
\begin{gather*}
\left( 2\pi i \hbar \right )^{-n/2} \int_{P_n\Hi}^o e^{\frac{i}{2\hbar} \|\gamma_n \|^2 }e^{-\frac{i}{\hbar}\int_0^t\ba(\gamma_n(s)+x)\dot\gamma_n ds}\psi_0(\gamma_n(t)+x)d \gamma _n= \\
=\left( 2\pi  \right )^{-n/2} \int_{P_n\Hi} e^{-\frac{1}{2} \|\gamma_n \|^2 }e^{-\frac{i\sqrt i}{\sqrt\hbar}\int_0^t\ba(\sqrt{i\hbar}\gamma_n(s)+x)\dot\gamma_n ds}\psi_0(\sqrt{i\hbar}\gamma_n(t)+x)d \gamma _n
\end{gather*}

and the last line is  equal to the r.h.s. of \eqref{fey-Wie-findim}, namely to:
$$
\E \left [\psi_0(\sqrt{ i\hbar}{\omega}_n(t)+x)e^{-\frac{i}{\hbar}\sqrt{ i\hbar}\int_0^t\ba (\sqrt{ i\hbar}{\omega}_n(s) +x) \cdot \dot{{\omega}}_n(s) ds} \right ].
$$
\end{proof}
The next step is the study of the convergence of the Wiener integrals on the r.h.s. of \eqref{fey-Wie-findim} which can be written as 
$$\E \left [\psi_0(\sqrt{ i\hbar}{\omega}_n(t)+x)e^{-\frac{i}{\hbar}\sqrt{i\hbar}\ba(x)\omega_n(t)}e^{g_n(\omega)} \right ],$$
where, given an orthonormal basis $\{e_n\}$ of $\Hi_t$, the random variables $g_n:C_t\to\C$ are defined by 
\begin{equation}
\label{n0}
g_n(\omega):=\int_0^t\ba (\omega_n(s)) \cdot \dot{{\omega}}_n(s) ds, \qquad \omega \in C_t,
\end{equation}
where $C_t$ is as in \eqref{Ct}.
Further, let us consider the sequence $\{\mathfrak{g}_n\}$ of real random variables on $(C_t, \Ba(C_t),\bP)$ defined as
\begin{equation}
\label{n1}
\mathfrak{g}_n(\omega):=\int_0^t\ba(\omega (t))\cdot \dot{{\omega}}_n(t)dt, \qquad \omega \in C_t,
\end{equation}
where ${\omega}_n$ is given by
$
{\omega}_n:=\tilde{P}_n(\omega)$ and $\tilde P_n$ is defined in \eqref{tildePn}.

Consider the linear operator  $\mathfrak{G}: C_t\to\Hi_t$ defined by
\begin{equation}
\label{functionG}
\mathfrak{G}(\omega)(s)=\int _0^s \ba(\omega (r))dr, \qquad \omega \in C_t, \ s\in [0,t],
\end{equation}
with its help the functions $\{\mathfrak{g}_n\}$ and $\{g_n\}$ can be represented by  the  inner products:
\begin{equation}
\label{gn2}
\mathfrak{g}_n(\omega)=\langle \mathfrak{G}(\omega), \tilde P_n(\omega)\rangle, \qquad g_n(\omega)=\langle \mathfrak{G}(\tilde P_n\omega), \tilde P_n(\omega)\rangle.
\end{equation}
For our purpose, it is useful to introduce the definition of $\mathcal{H}$-differentiable function, following, e.g.,  \cite{RAMER}:
\begin{definition} 
A function ${\mathfrak{G}}:C_t\to C_t$ with ${\mathfrak{G}}(C_t)\subset \Hi_t$  is said to be {\em $\Hi_t$-differentiable} if for any $\omega\in C_t$ the function ${\mathfrak{G}}_\omega:\Hi_t\to \Hi_t$ defined as ${\mathfrak{G}}_\omega(\gamma)={\mathfrak{G}}(\omega +\gamma)$, $\gamma\in \Hi_t$, is Fr\'echet differentiable at the origin in $\Hi_t$. Its Fr\'echet derivative, namely the linear operator ${\mathfrak{DG}}_\omega(0)\in L(\Hi_t;\Hi_t)$,  will be denoted with the symbol $\mathfrak{DG}(\omega) $ and called the {\em $\Hi_t$-derivative of ${\mathfrak{G}}$ at $\omega$.}
\end{definition}


\begin{lemma}
\label{Lemma3}
Let $\mathfrak{G}: C_t\to \Hi_t$ be a linear operator such that its restriction $\mathfrak{G}_{\Hi_t}$ on $\Hi_t$ is Hilbert-Schmidt. Let $\{P_n\}_n$ be a sequence of finite dimensional projection operators in $\Hi_t$ converging strongly to the identity. Then the sequences of random variables $\{\mathfrak{g}_n\}$ and $\{\mathfrak{g}'_n\}$ on $C_t$ defined as:
\begin{eqnarray*}
\mathfrak{g}_n(\omega)&=&\langle \mathfrak{G}(\omega), \tilde P_n(\omega)\rangle, \qquad \omega \in C_t,\\
\mathfrak{g}'_n(\omega)&=&\langle \mathfrak{G}(\tilde P_n(\omega)), \tilde P_n(\omega)\rangle, \qquad \omega \in C_t,
\end{eqnarray*}
satisfy
\begin{equation}\label{conv-eq}
\lim_{n\to \infty}\bE[|\mathfrak{g}_n-\mathfrak{g}'_n|^2]=0.
\end{equation}

\end{lemma}
\begin{proof}
\begin{eqnarray*}
\bE[|\mathfrak{g}_n-\mathfrak{g}'_n|^2]&=&\int |\langle \mathfrak{G}(\omega)-G(\tilde P_n(\omega)), \tilde P_n(\omega)\rangle|^2d\bP(\omega)\\
&=&\int |\langle \mathfrak{G}(\omega-\tilde P_n(\omega)), \tilde P_n(\omega)\rangle|^2d\bP(\omega)\\
&=&\int |\langle \mathfrak{G}(\sum_{j=n+1}^\infty e_j n_{e_j}(\omega)),\sum_{i=1}^n e_i n_{e_i}(\omega)\rangle|^2d\bP(\omega)\\
&=&\sum_{j,j'=n+1}^\infty \sum_{i,i'=1}^n \langle \mathfrak{G}e_j,e_i\rangle \langle \mathfrak{G}e_{j'},e_{i'}\rangle \bE [n_{e_j}n_{e_{j'}}n_{e_i}n_{e_{i'}}]\\
&=&\sum_{j=n+1}^\infty \sum_{i=1}^n( \langle \mathfrak{G}e_j,e_i\rangle)^2 \\
&=&\sum_{j=n+1}^\infty \langle P_n\mathfrak{G}e_j,P_n\mathfrak{G}e_j\rangle,
\end{eqnarray*}
where in the third step we have applied It\=o-Nisio theorem.
By using the assumption that $\mathfrak{G}_\Hi$ is an Hilbert-Schmidt operator we obtain \eqref{conv-eq}.

\end{proof}
In this setting and with $\ba$ given by \eqref{vector-a} the map $\mathfrak{G}:C \to \mathcal{H}_t$, as defined by \eqref{functionG}, is given by
\begin{gather*}
\mathfrak{G}(\omega)(s)=\Big (
\alpha^1_1 \mathcal{I}(\omega_1) +
\alpha^1_2 \mathcal{I}(\omega_2)+
\alpha^1_3 \mathcal{I}(\omega_3),
\alpha^2_1 \mathcal{I}(\omega_1)+
\alpha^2_2 \mathcal{I}(\omega_2)+
\alpha^2_3 \mathcal{I}(\omega_3),\\
\alpha^3_1 \mathcal{I}(\omega_1)+
\alpha^3_2 \mathcal{I}(\omega_2)+
\alpha^3_3 \mathcal{I}(\omega_3)
\Big )
\end{gather*}
where 
\[
\mathcal{I}(\omega_k)=\int_0^s \omega_k(r) dr, \qquad k=1,2,3,
\]
and $\omega=(\omega_1,\omega_2,\omega_3) \in C_t$. Thus, in this case, the $\Hi_t-$derivative ${\mathfrak{DG}}(\omega)$ for any $\omega \in C_t$ is the linear operator given by $\mathfrak{DG}_{\omega}(\gamma)=G\gamma$, where $G:\Hi_t\to\Hi_t$ is defined in \eqref{op-G}.
In particular, according to lemmas \ref{Lemma3} and \ref{lemmaG} the sequences of random variables $\{g_n\}$ and $\{\mathfrak{g}_n\}$ defined respectively by \eqref{n0} and \eqref{n1} share the same convergence properties in $L^2(C_t, \bP)$.

The following result is a direct consequence of Lemmas 4.2 and 4.3 in \cite{RAMER}.
\begin{teorema}
\label{RAMERFORM}
Let ${\mathfrak{G}}:C_t\to C_t$, with ${\mathfrak{G}}(C_t)\subset \Hi_t$,  be a $\Hi_t$-differentiable map such that for any $\omega \in C_t$ the $\Hi_t$-derivative ${\mathfrak{DG}}(\omega)\in L(\Hi_t,\Hi_t)$ is an Hilbert-Schmidt operator. Let us assume furthermore that the maps $\|{\mathfrak{G}}\|:C_t\to \mathbb{R}$ and $\|{\mathfrak{DG}}\|_2:C_t\to \mathbb{R}$, where  $\|{\mathfrak{DG}}(\omega)\|_2$ denotes the Hilbert-Schmidt norm of ${\mathfrak{DG}}(\omega)$, belong to $L^2(C_t, \mathbb{P})$. Let $\{e_i\}$ be an orthonormal basis of $\Hi_t$ and let $\{P_n\}$ and $\{\tilde P_n\}$ be the sequence of finite dimensional projectors on the span of $e_1,\dots,e_n$ and their stochastic extensions to $C_t$ respectively. Then the sequence of random variables $\{ \mathfrak{h}_n\}$ defined as 
$$\mathfrak{h}_n(\omega):= \langle {\mathfrak{G}}(\omega), \tilde P_n(\omega)\rangle -\Tr (P_n{\mathfrak{DG}}(\omega)), \qquad \omega\in C_t,$$
converges in $L^2(C_t,  \mathbb{P})$ and the limit does not depend on the basis $\{e_i\}$, $i=1,\dots , n$.
\end{teorema}
 The previous theorem applied to our particular case provides actually a no-go result on the convergence of the sequence of random variables $\{g_n\}$ given by \eqref{n0}. Indeed, since in our particular case ${\mathfrak{DG}}(\omega)=G$ and by Remark \ref{rem-notrace} the operator $G:\Hi_t\to \Hi_t$ is not trace class, the sequence of real numbers $\Tr (P_n{\mathfrak{DG}}(\omega))\equiv \Tr (P_nG)$ does not in general converge independently on the choice of finite dimensional approximations $\{P_n\}$.  Hence, by Theorem \ref{RAMERFORM} and Lemma \ref{Lemma3} neither the random variables $\{\mathfrak{g}_n\}$ nor $\{g_n\}$ admit a well defined limit in $L^2(C_t, \bP)$ independent on the sequence $\{P_n\}$ of finite dimensional projectors. 
 The following theorem provides a suitable renormalization term, namely a sequence $\{r_n\}$ of real numbers such that the renormalized random variables $h_n:C_t\to \R$ given by $h_n(\omega ):=g_n (\omega)-r_n$ converge in $L^p(C_t, \bP)$ for all $p\geq 1$ and in probability to the Stratonovich stochastic integral $h(\omega)=\int_0^t\ba (\omega (s))\circ d\omega (s)$.

\begin{teorema}
\label{const3D}
Let $\ba:\R^3\to \R^3$ be a linear vector field, let $\{e_k\}$ be an orthonormal basis of $\Hi_t$ and let $\{P_n\}$ and $\{\tilde P_n\}$ be the sequence of finite dimensional projectors on the span of $e_1,\dots,e_n$ in $\Hi_t$ and their stochastic extensions to $C_t$ respectively. Then by setting \begin{equation}\label{rn}
r_n:=\mathbf{B}\cdot\frac{1}{2}\sum_{k=1}^n \int_0^t e_k (s)\wedge \dot{e}_k (s) ds
\end{equation} with $\mathbf{B}=\rot \ba$, the sequence of random variables $h_n:C_t \to \R$ defined as:
\[
h_n(\omega):=\int_0^t \mathbf{a}(\omega_n(s) )\cdot \dot{\omega_n}(s) ds - r_n, \quad \omega \in C_t, 
\]
where $\omega_n := \tilde P_n(\omega)$, 
converges in $L^2(C_t, \mathbb{P})$, independently of $\{P_n\}_n$, to
\[
\int_0^t \mathbf{a} (\omega(s)) \circ d\omega(s).
\]
\end{teorema}
\begin{proof}
Let us set
\begin{align*}
X_n(\omega)&= \int_0^t \mathbf{a}(\omega_n(s))\cdot \dot{\omega}_n(s) ds\\
&= \int_0^t a_1(\omega_n(s)) \dot{\omega}_{n,1}(s) \ ds + \int_0^t a_2(\omega_n(s))\dot{\omega}_{n,2} (s)ds +\int_0^t a_3(\omega_n(s)) \dot{\omega}_{n,3} (s)ds
\end{align*}
where $\omega_n=\left ( {\omega}_{n,1}, {\omega}_{n,2}, {\omega}_{n,3} \right )\in \mathcal{H}_t$.
By Stokes theorem:
\begin{equation}\label{stokes}
X_n=\iint_{S_n} \mathbf{B} \cdot \mathbf{n} \  dS -\int_{\Lambda_n} \mathbf{a}  \cdot d {\bf r}
\end{equation}
where $\Lambda_n$ is the (oriented) segment joining $\omega_n(t)$ with 0, while $\int_{\Lambda_n} \mathbf{a}  \cdot d {\bf r}$ is the line integral of $\ba$ along $\Lambda _n$. $S_n$  is any regular oriented surface with oriented boundary given by the close path union of $\omega_n$ and $\Lambda_n$, $\mathbf{n}$ denotes the normal unit vector and $\iint_{S_n} \mathbf{B} \cdot \mathbf{n} \  dS $ is the surface integral of $\bB$ on $S_n$. Our study can be restricted to 
\[
\iint_{S_n}\mathbf{B}\cdot \mathbf{n}  \ dS,
\]
as we can immediately see that the second term converges in $L^2(C_t,\bP) $ independently of $\{P_n\}$. Indeed:
\[
\int_{\Lambda_n} \mathbf{a} \cdot d\mathbf{r}=\int_0^1 \mathbf{a}(u\omega_n(t))du\cdot \omega_n(t)  ,\]
and for any sequence of finite dimensional projection operators $\{P_n\}_n$ such that $P_n \to \mathbb{I}$ we have
\[
\omega_n(t) \to \omega(t), \qquad \int_0^1\ba(u \,\omega_n(t))du\to \int_0^1\ba(u\omega(t))du, \qquad \forall t\geq 0.
\]
Let us consider now the surface integral $\iint_{S_n} \mathbf{B} \cdot \mathbf{n} \, dS$. Since by the assumption on $\ba$ the magnetic field $\mathbf{B}$ is constant, by the Gauss-Green formula we get
$$
\iint_{S_n}\mathbf{B}  \cdot \mathbf{n}  \, dS = \mathbf{B} \iint_{S_n}\mathbf{n}   \ dS=
\mathbf{B}\cdot \frac{1}{2}\int_0^t \omega_n(s) \wedge \dot{\omega}_n (s)ds.
$$
Let us define for any $i=1,2,3$ the sequence of random variables $h_n^i:C_t\to\R$ by
\begin{equation*}
h_n^i(\omega):=\hat e_i\cdot \int_0^t  \omega_n (s) \wedge \dot{\omega}_n (s) \ ds=\langle H^i(\omega_n), \omega_n\rangle, \qquad \omega \in C_t,
\end{equation*}
where  $\hat e_i$, $i=1,2,3$, are the vectors of the canonical basis of $\R^3$ and  the linear operators $H^i:C_t\to \Hi_t$ are defined by $$(H^i(\omega)(s))^T:=\int_0^s J^i\omega (u)^Tdu,$$
with $\,^T$ denoting the transpose and $J^i$, $i=1,2,3$,  are the matrices:
$$
J^1=\left(\begin{array}{lll} 0 & 0&0\\ 
0 & 0 & -1\\ 0 & 1 & 0   \end{array}\right), \qquad J^2=\left(\begin{array}{lll} 0 & 0&1\\ 
0 & 0 & 0\\ -1 & 0 & 0   \end{array}\right),\qquad J^3=\left(\begin{array}{lll} 0 & -1&0\\ 
1 & 0 & 0\\ 0 & 0 & 0   \end{array}\right). 
$$
Actually the operators $H^i$, $i=1,\dots,3$ have the form \eqref{op-G} and by Lemma \ref{lemmaG} are Hilbert-Schmidt. Further, by Lemma \ref{Lemma3} and Theorem \ref{RAMERFORM}, the renormalized  random variables
$$    h_n^i(\omega)-r_n^i=\langle H^i(\omega_n), \omega_n\rangle- \Tr[P_nH^i]=\int_0^t  (\gamma_n (s)\wedge \dot{\gamma}_n(s))_i  \ ds-\sum_{k=1}^n \left ( e_k (s)\wedge \dot{e}_k(s) \right )_i ds$$
converge in $L^2(C_t,\bP)$  and the limit does not depend on the sequence $\{P_n\}$. By combining these results we obtain the convergence of the sequence 
$$\int_0^t \mathbf{a}(\omega_n(s) )\cdot \dot{\omega_n}(s) ds -\mathbf{B}\cdot\frac{1}{2}\sum_{k=1}^n \int_0^t e_k(s) \wedge \dot{e}_k (s) ds$$ and the limit is independent of  the sequence $\{P_n\}$.
Eventually, by choosing the sequence $\{P_n\}$ of piecewise linear approximations \eqref{gamman}, where the elements $e_n$ of the corresponding basis $\{e_n\}$ satisfy $\int_0^te_n (s)\wedge \dot e_n (s)ds=0$, and by applying lemma \ref{LemmaAppendix} we complete the proof.
\end{proof}
\begin{Remark} 
The renormalization term given in Theorem \ref{const3D} contains, besides the magnetic field $\bB$, the area integrals of the elements $\{e_n\}$ of the orthonormal basis spanning the finite dimensional Hilbert space $P_n\Hi_t$.
This term is gauge independent. However, for a general orthonormal basis in $\Hi_t$, it does not converge to a well defined limit as the following example shows.

Let us fix $t\equiv 1$ and let us consider two sequences of real valued functions $\{u_n\}_{n\geq 0}$, $\{v_n\}_{n\geq1}$ defined on the interval $[0,1]$ by 
$u_0(s)=s $ and for $n\geq 1$  $$ u_n(s)= \frac{\cos(2\pi n s)}{2\pi n}, \quad 
v_n(s)= \frac{\sin(2\pi n s)}{2\pi n}, \qquad s\in [0,1]$$
and the sequences of vectors in $\Hi_t$ defined by:
$e_{n,1}:=(u_n,v_n,0)$,  $e_{n,2}:=(u_n,-v_n,0)$, $e_{n,3}:=(v_n,u_n,0)$, $e_{n,4}:=(v_n,-u_n,0)$, $e_{n,5}:=(0,0,u_n)$, $e_{n,6}:=(0,0,v_n)$, 
that together  with the vectors $e_{0,1}:=(u_0,0,0)$, $e_{0,2}:=(0,u_0,0)$ and $e_{n,3}:=(0.0, u_0)$ provide an orthonormal basis of $\Hi_t$.
Given the linear the vector field $\ba:\R^3\to\R^3$ \[
\mathbf{a}(x,y,z)=\left ( \frac{z-y}{2}, \frac{x-z}{2},\frac{y-x}{2}\right ).
\] 
with $\mathbf{B}=\rot \mathbf{a}=(1,1,1)$, and 
taking the vectors $e_{k,1}$ and  $e_{k,4}$, $k=1,\dots, n$, we get
\[
\mathbf{B} \cdot \frac{1}{2}\sum_{k=1}^n \int_0^t e_{k,1} \wedge \dot{e}_{k,1} =\mathbf{B} \cdot \frac{1}{2}\sum_{k=1}^n \int_0^t e_{k,4} \wedge \dot{e}_{k,4} =\sum_{k=1}^n\frac{1}{4\pi k}.
\]
On the other hand, considering the vectors $e_{k,2}$ and  $e_{k,3}$, $k=1, \dots, n$, we have
\[
\mathbf{B} \cdot \frac{1}{2}\sum_{k=1}^n \int_0^t e_{k,2} \wedge \dot{e}_{k,2} =\mathbf{B} \cdot \frac{1}{2}\sum_{k=1}^n \int_0^t e_{k,3} \wedge \dot{e}_{k,3} =-\sum_{k=1}^n \frac{1}{4\pi k},
\]
while the other vectors of the orthonormal basis give vanishing area integrals. Hence, the renormalization term $r_n$ given by \eqref{rn} is not absolutely convergent as $n \to \infty$.
\end{Remark}
A direct consequence of Lemma \ref{lemma-anal-linear} and Theorem \ref{const3D} is the following result.
\begin{corollary}
\label{Cor1}
Under the assumptions of Lemma \ref{lemma-anal-linear}, the sequence of finite dimensional renormalized oscillatory integrals
$$ 
 (2\pi i \hbar)^{-n/2}\widetilde{\int_{P_n\Hi}} e^{\frac{i}{2\hbar} \| \gamma_n \|^2} e^{-\frac{i}{\hbar}\left(\int_0^t \ba (\gamma_n(s)\cdot \dot \gamma_n (s)ds-r_n\right)} \psi_0(\gamma_n(t)+x)d\gamma_n,$$ with the renormalization term $r_n$ given by \eqref{rn}, converges as $n\to \infty $ to the Wiener integral 
 \begin{equation}
 \label{Wfin}
 \mathbb{E} \left [ \psi_0(\sqrt{i\hbar}\omega(t)+x)e^{-\frac{i}{\hbar}\sqrt{i\hbar}\int_0^t\ba(\sqrt{i\hbar}\omega(s)+x)\circ d\omega(s)}\right ]
\end{equation}
and the limit is independent of the sequence $\{P_n\}$ of finite dimensional approximations. In addition, it provides the solution of the Schr\"odinger equation with magnetic field
\begin{equation}
\label{Schfin}
\left\{\begin{array}{l}
i\hbar \partial_t \psi(t,x)=\frac{1}{2}\left (-i\hbar \nabla - \ba\right )^2\psi(t,x)\\
\psi(0,x)=\psi_0(x)
\end{array}
\right.
, \qquad t \in \R^+, \quad x \in \R^3.
\end{equation}
\end{corollary}

\begin{proof}
The first part of the theorem follows from Lemma \ref{lemma-anal-linear} and Theorem \ref{const3D}. The second part can be proved by using the analyticity properties of the semigroup generated by the quantum Hamiltonian operator $H=\frac{1}{2}(-i\hbar \nabla -\ba)^2$. More precisely
for $t\in \R^+$ the action of the heat semigroup  on the vector $\psi _0$ is given by the Feynman-Kac-It\=o formula:
$$
e^{-\frac{t}{\hbar}H}\psi_0(x)=\bE \left [\psi_0(\sqrt \hbar \omega (t)+x)e^{-\frac{i}{\sqrt\hbar}\int_0^t\ba (\sqrt\hbar \omega (s)+x)\circ d\omega (s)}\right ].
$$
For any $\phi\in L^2(\R^3)$ the inner product $\langle \phi, e^{-z\frac{t}{\hbar}H}\psi_0\rangle$ is an analytic function of $z\in D$, $D=\{z\in \C, Re(z)\geq 0\}$,  continuous on $\bar D$, giving for $z=i$ the inner product between $\phi \in L^2(\R^d)$ and the solution of the Schr\"odinger equation \eqref{Schfin}. For $z\in \R^+$, by the change of variables formula we have 
$$\langle \phi, e^{-z\frac{t}{\hbar}H}\psi_0\rangle=\int_{\R^3}\bar \phi (x)\bE \left [\psi_0(\sqrt{z \hbar} \omega (t)+x)e^{-\frac{i\sqrt {z}}{\sqrt\hbar}\int_0^t\ba (\sqrt{z\hbar} \omega (s)+x)\circ d\omega (s)} \right ]dx.$$
By the assumptions on $t, \ba$, and $\psi_0$, both sides of the equality above are analytic for $z\in D$, continuous in $\bar D$ and  coincide on $\R^+$. Hence, for $z=i$ we obtain that the solution in $L^2(\R^3)$ of \eqref{Schfin} is given by \eqref{Wfin}.
\end{proof}

\begin{Remark}\label{RemDysonV-2}
The results of Corollary \ref{Cor1} can be  generalized to the case where a scalar potential $V\in \Fo_c(\R^3)$ is added to the Hamiltonian, i.e. $H=H_0+V$ with $H_0=\frac{1}{2}(-i\hbar \nabla -\ba)^2$. Indeed, in this case, since the function $V:\R^3\to \R$ has the form \eqref{hypV}, it is bounded and can be extended to an analytic function $V:\C^3\to \C$. It is easy to verify that the multiplication operator associated with $V$ is bounded, i.e. $\|V\psi \|\leq \sup_{x\in \R^3}|V(x)|\|\psi\|_{L^2(\R^3)}$ hence the perturbative Dyson expansion for the vector $e^{-\frac{it}{\hbar}(H_0+V)}\psi_0$  is convergent. We have 
$$e^{-\frac{it}{\hbar}(H_0+V)}\psi_0=\sum_m \left(-\frac{i}{\hbar}\right)^m \phi_m,$$ where 
$$\phi_m=\int_{\Delta_m(t)}e^{-\frac{i}{\hbar}H_0(t-s_m)}Ve^{-\frac{i}{\hbar}H_0(s_m-s_{m-1})}\cdots Ve^{-\frac{i}{\hbar}H_0(s_2-s_1)}Ve^{-\frac{i}{\hbar}H_0s_1}\psi_0 ds_1\dots ds_m$$
with  $\Delta_m(t)=\{(s_1, \dots, s_m ) \in \R^m \colon 0\leq s_1\leq \dots\leq s_m\leq t\}$

By exploiting the analyticity for the semigroup generated by $H_0$,   the Dyson expansion for heat semigroup $e^{-\frac{t}{\hbar}(H_0+V)}\psi_0$ as well as the techniques used in the proof of Corollary \ref{Cor1}, it is simple to prove that for any $m\in \N$ the vector $\phi_m$ can be represented in terms of the limit  of the following sequence of finite dimensional renormalized oscillatory integrals:
$$ \phi_m(t,x)=\lim_{n\to \infty}
 (2\pi i \hbar)^{-n/2}\widetilde{\int_{P_n\Hi}} e^{\frac{i}{2\hbar} \| \gamma_n \|^2} e^{-\frac{i}{\hbar}\left(\int_0^t \ba (\gamma_n(s)\cdot \dot \gamma_n (s)ds-r_n\right)}\left(\int_0^t V(\gamma_n(s)+x)ds\right)^m \psi_0(\gamma_n(t)+x)d\gamma_n.$$
The limit is independent on the choice of the sequence $\{P_n\}$ of finite dimensional approximations and it is equal to the Wiener integral
$$ \mathbb{E} \left [ \psi_0(\sqrt{i\hbar}\omega(t)+x)\left(\int_0^t V(\sqrt{i\hbar}\omega(s)+x)ds\right)^me^{-\frac{i}{\hbar}\sqrt{i\hbar}\int_0^t\ba(\sqrt{i\hbar}\omega(s)+x)\circ d\omega(s)}\right ].$$
\end{Remark}
\begin{Remark}
Similarly as in Remark \ref{Rem-10} we point out that all the results in Sect. \ref{sez4} can be extended to the case where the space dimension $d$ is arbitrary.
\end{Remark}

\section{Acknowledgements}
The first named author is very grateful to Elisa Mastrogiacomo and Stefania Ugolini for invitations to University of Insubria, Varese, and Universit\'a degli studi, Milano,  that greatly facilitated our scientific cooperation. Also our participations to workshops in Trento, organized by Stefano Bonaccorsi and Sonia Mazzucchi, in 2017, and in Rome, organized by Alessandro Teta, in 2018 gave us an excellent opportunity of advancing our joint research and we are very grateful for these opportunities. The third named author gratefully acknowledges the hospitality of the Hausdorff Center and the University of Bonn, as well as the support of the Alexander von Humboldt Stiftung.  


\section*{Appendix A: proof of Lemma \ref{LemmaAppendix}}


Let us consider the sequence of random variables $\{g_n\}$ defined by 
$$
g_n (\omega)= \sum_{j=0}^{n-1} a(\sqrt{i\hbar}\omega(s_j))\left ( \omega(s_{j+1})-\omega(s_j) \right ), \qquad \omega \in C_t
$$
and the stochastic integral$$
G (\omega)= \int_0^t a(\sqrt{i\hbar}\omega(s)) d\omega(s),$$
where $s_j=\frac{jt}{n}$ and $a$ is the Fourier transform of a complex bounded measure on $\R$ with compact support contained in the ball $B_R$ with radius $R \in \R^+$:
\[
a(\sqrt{i\hbar}\omega(s))=\int_{\mathbb{R}}e^{i\sqrt{i\hbar}\xi\omega(s)}d\mu(\xi).
\] 
Without loss of generality, we can restrict ourselves to prove the convergence of $g_n$ to $G$ in $L^p(C_t, \bP)$ for $p$ even. 

By BDG inequalities (see ,e.g. \cite{KarSh}) we have
\begin{equation}
\label{E1}
\mathbb{E} \left [ | G-g_n |^{2p} \right] \le C_{2p} \cdot  \mathbb{E} \left [ \left ( \sum_{j=0}^{n-1} {\int_{s_j}^{s_{j+1}}} \left | a(\sqrt{i\hbar}\omega(s))-a(\sqrt{i\hbar}\omega(s_j))\right |^2 ds\right )^p \right ],
\end{equation}
with $C_{2p}$ a positive constant. Moreover:
\begin{align}
\left | a(\sqrt{i\hbar}\omega(s))-a(\sqrt{i\hbar}\omega(s_j))\right |^2 =& ( \omega(s)-\omega(s_j))^2 \left | \int_0^1 \sqrt{i\hbar}a'\left(\sqrt{i\hbar}\left(\omega(s_j)+u(\omega(s)-\omega(s_j))\right)\right)du \right |^2 \nonumber \\
=&  ( \omega(s)-\omega(s_j))^2 \left | i\sqrt{i\hbar} \int_0^1 \int_{\R} \xi e^{i\sqrt{i\hbar}\xi(\omega(s_j)+u(\omega(s)-\omega(s_j)))} d\mu(\xi) du \right |^2  \label{App1}\\
\leq &\hbar ({\omega(s)}-{\omega(s_j)})^2 \cdot \mathcal{G}(\omega(s),{\omega(s_j)}), \nonumber
\end{align}
where 
\begin{gather}
\mathcal{G}(\omega(s),{\omega(s_j)})= \nonumber\\
=\int_0^1 \int_0^1 \int_{\R}\int_{\R}
|\xi_1||\xi_2|e^{-\frac{\sqrt{2}}{2}\xi_1({\omega(s_j)}+u_1({\omega(s)}-{\omega(s_j)}))} 
 e^{-\frac{\sqrt{2}}{2}\xi_1({\omega(s_j)}+u_2({\omega(s)}-{\omega(s_j)}))}d|\mu|(\xi_1)d|\mu|(\xi_2)du_1du_2 \label{DefG}
 \end{gather}
Using \eqref{App1}, can rewrite the expectation (\ref{E1}) as follows
\begin{align}
\E \left [ \left | G-g_n \right |^{2p}\right ] \le& C_{2p} \hbar ^p \E \left(\sum_{j=0}^{n-1} {\int_{s_j}^{s_{j+1}}} ({\omega(s)}-{\omega(s_j)})^2\mathcal{G}(\omega(s),{\omega(s_j)})ds \right)^p\nonumber\\
=& C_{2p} \hbar ^p \cdot  \E  \left [ \sum_{j_1,\dots, j_p=0}^{n-1} {\int_{s_{j_1}}^{s_{j_1+1}}}\cdots {\int_{s_{j_p}}^{s_{j_p+1}}} {(\omega(s_1)-\omega(s_{j_1}))}^2 \cdots {(\omega(s_p)-\omega(s_{j_p}))}^2 \right .\nonumber\\
 & \qquad  {\mathcal{G}(\omega(s_1),\omega(s_{j_1}))} \cdots {\mathcal{G}(\omega(s_p),\omega(s_{j_p}))}  ds_1 \cdots ds_p  \vast ]\nonumber\\
 \leq &C_{2p}\hbar ^p I_n^1 \, I_n^2, \label{69bis}
\end{align}
where in the latter inequality we used Schwarz inequality, with
\begin{align}
I_n^1=&\sqrt{\E \left [ {\sum_{j_1,\dots, j_p=0}^{n-1}} \left ( {\int_{s_{j_1}}^{s_{j_1+1}}}\cdots {\int_{s_{j_p}}^{s_{j_p+1}}} {(\omega(s_1)-\omega(s_{j_1}))}^4 \cdots {(\omega(s_p)-\omega(s_{j_p}))}^4 ds_1 \cdots ds_p \right ) \right ]},\label{I1n}
\\
I_n^2=& \sqrt{\E \left [ {\sum_{j_1,\dots, j_p=0}^{n-1}} \left ( {\int_{s_{j_1}}^{s_{j_1+1}}}\cdots {\int_{s_{j_p}}^{s_{j_p+1}}}  {\mathcal{G}(\omega(s_1),\omega(s_{j_1}))}^2 \cdots {\mathcal{G}(\omega(s_p),\omega(s_{j_p}))}^2 ds_1 \cdots ds_p \right )\right ]}.\label{I2n}
\end{align}
We will show that $I_n^1 \to 0$ for $n \to \infty$ and that $I_n^2$ is bounded for all $n$.

Let us consider the integral  $I_n^1$ given by \eqref{I1n}.

All the expectations $\E \left [ {(\omega(s_1)-\omega(s_{j_1}))}^4 \cdots {(\omega(s_p)-\omega(s_{j_p}))}^4 \right ]$ can be computed taking into account the  coincidences of the indices $j_r$, with $r=1,\dots, p$ in the following way:
\begin{gather} 
{\int_{s_{j_1}}^{s_{j_1+1}}}\cdots {\int_{s_{j_p}}^{s_{j_p+1}}} \E \left [ {(\omega(s_1)-\omega(s_{j_1}))}^4 \cdots {(\omega(s_p)-\omega(s_{j_p}))}^4 \right ] ds_1 \cdots ds_p= \nonumber\\
=\int_0^{\frac{t}{n}} \cdots  \int_0^{\frac{t}{n}} \E \left [ \omega(s_1)^4 \cdots \omega(s_{p_1})^4 \right ] ds_1 \cdots ds_{p_1}
\int_0^{\frac{t}{n}} \cdots  \int_0^{\frac{t}{n}} \E \left [ \omega(s_{p_1+1})^4\cdots \omega(s_{p_1+p_2})^4\right ]ds_{{p_1}+1} \cdots ds_{p_1+p_2} \nonumber \\
\qquad \dots \qquad \int_0^{\frac{t}{n}} \cdots  \int_0^{\frac{t}{n}} \E \left[\omega(s_{p_{i-1}+1})^4 \cdots \omega(s_{p_{i-1}+p_i})^4 \right ] ds_{p_{i-1}+1} \cdots ds_{p_{i-1}+p_i}, \label{intop}
\end{gather}
with $p_1+p_2+ \cdots + p_i=p$ and we have used that in distribution ${\omega(s)}-{\omega(s_j)} \sim \omega(s-s_j)$.
Further, the generic term containing $\tilde{p}$ factors, for any $\tilde{p}=1, \dots, p$, can be computed as
\begin{equation*}
\int_0^{\frac{t}{n}} \cdots  \int_0^{\frac{t}{n}} \E \left[ \omega(s_1)^4 \cdots \omega(s_{\tilde{p}})^4 \right] ds_1 \cdots ds_{\tilde{p}}= \tilde{p}! \idotsint\displaylimits_{0<s_1<\cdots<s_{\tilde{p}}<t/n} \E \left [  \omega(s_1)^4 \cdots \omega(s_{\tilde{p}})^4\right ]ds_1 \cdots ds_{\tilde{p}}.
\end{equation*}
By a straightforward calculation\footnote{
The first step it is given by:
\begin{align*}
\E \left [ \omega(s_1)^4\cdots \omega(s_{\tilde{p}})^4\right ]&=\E \left [ \omega(s_1)^4\cdots (\omega(s_{\tilde{p}})-\omega(s_{\tilde{p}-1})+\omega(s_{\tilde{p}-1}))^4 \right]\\
&= \E \left [\omega(s_1)^4\cdots \omega(s_{\tilde{p}-1})^8 \right]+\E \left [  \omega(s_1)^4\cdots \omega(s_{\tilde{p}-1})^6(\omega( s_{\tilde{p}})-\omega( s_{\tilde{p}-1}))^4\right]\\
& \ + 6\cdot \E \left [ \omega(s_1)^4\cdots \omega(s_{\tilde{p}-1})^6(\omega( s_{\tilde{p}})-\omega( s_{\tilde{p}-1}))^2\right],
\end{align*}
then we proceed in the same way for $\tilde{p}$ steps.
} we can represent $\E \left [ \omega(s_1)^4 \cdots \omega(s_{\tilde{p}})^4 \right ]$ as a homogeneous polynomial
$P(s_1,s_2-s_1, \dots, s_{\tilde{p}}-s_{\tilde{p}-1})$ with $\deg(P)=2\tilde{p}$. We can rewrite it as $Q(s_1,s_2,\dots, s_{\tilde{p}})$, with $\deg(Q)=2\tilde{p}$ (its coefficients depending only on $\tilde{p}$).  Thanks to the change of variables $t_i=\frac{s_i}{t/n}$, we have
\[
\idotsint\displaylimits_{0<s_1<\cdots<s_{\tilde{p}}<t/n}Q(s_1,s_2,\dots, s_{\tilde{p}})ds_1\cdots ds_{\tilde{p}}=\idotsint\displaylimits_{0<t_1<\cdots<t_{\tilde{p}}<1} \left(\frac{t}{n} \right)^{3\tilde{p}} Q(t_1,t_2,\dots, t_{\tilde{p}}) dt_1 \cdots dt_{\tilde{p}}= C \cdot \left (\frac{t}{n} \right )^{3\tilde{p}},
\]
with 
$$
C=\int_{0<t_1< \cdots < t_{\tilde{p}}<1} Q(t_1,t_2,\dots, t_{\tilde{p}}) dt_1 \cdots dt_{\tilde{p}}.
$$ 
Applying the same argument for all terms in \eqref{intop} we get
\begin{align*}
{\int_{s_{j_1}}^{s_{j_1+1}}}\cdots {\int_{s_{j_p}}^{s_{j_p+1}}} \E \left [ {(\omega(s_1)-\omega(s_{j_1}))}^4 \cdots {(\omega(s_p)-\omega(s_{j_p}))}^4 \right ] ds_1 \cdots ds_p
&=C_1 \cdots C_i \cdot \left(\frac{t}{n}\right )^{3(p_1+\cdots+p_i)}\\
&=\widetilde{C}\cdot\left(\frac{t}{n}\right)^{3p},
\end{align*}
with $\tilde{C}=C_1 \cdots C_i$.
Thus all the contributions can be estimated  by $\widetilde{K}_p \cdot \left(\frac{t}{n}\right)^{3p}$, where $\widetilde{K}_p$ is the maximum of the constants computed as $\widetilde{C}$. Eventually,  using ${\sum_{j_1,\dots, j_p=0}^{n-1}} 1 =n^p$ we get
\[
I_n^1 \le \sqrt{{\sum_{j_1,\dots, j_p=0}^{n-1}} \widetilde{K}_p\cdot \left(\frac{t}{n}\right)^{3p}}=\sqrt{\left(\frac{t}{n}\right)^{3p}\cdot \widetilde{K}_p \cdot n^p}=\widetilde{\mathcal{K}}_p \cdot \frac{t^{\frac{3p}{2}}}{n^p} \xrightarrow{n \to \infty} 0.
\]

Concerning $I_n^2$, recalling the definition \eqref{DefG} of $\mathcal{G}$,  we have to study
$$
 I_n^2=\sqrt{{\sum_{j_1,\dots, j_p=0}^{n-1}} \left ( {\int_{s_{j_1}}^{s_{j_1+1}}}\cdots {\int_{s_{j_p}}^{s_{j_p+1}}}  \E \left [ {\mathcal{G}(\omega(s_1),\omega(s_{j_1}))}^2 \cdots {\mathcal{G}(\omega(s_p),\omega(s_{j_p}))}^2 \right ] ds_1 \cdots ds_p \right )}.
 $$

By writing explicitly the functions $\mathcal{G}(\cdot,\cdot)$, we get the following bound:
\begin{gather*}
\left ( I_n^2 \right )^2\le {\sum_{j_1,\dots, j_p=0}^{n-1}} \left ( {\int_{s_{j_1}}^{s_{j_1+1}}}\cdots {\int_{s_{j_p}}^{s_{j_p+1}}} \int_0^1 \cdots \int_0^1 \int_{\R} \cdots \int_{\R} \E \left [ \prod_{i=1}^p |\xi_i||\tilde{\xi}_i||\zeta_i||\tilde{\zeta}_i| \right . \right .\\
e^{-\frac{\sqrt{2}}{2}\xi_i(\omega(s_{j_i})+u_i(\omega(s_i)-\omega(s_{j_i}))} \cdot 
e^{-\frac{\sqrt{2}}{2}\tilde{\xi}_i(\omega(s_{j_i})+\tilde{u}_i(\omega(s_i)-\omega(s_{j_i}))} \cdot 
e^{-\frac{\sqrt{2}}{2}\zeta_i(\omega(s_{j_i})+v_i(\omega(s_i)-\omega(s_{j_i}))}
\\
e^{-\frac{\sqrt{2}}{2}\tilde{\zeta}_i(\omega(s_{j_i})+\tilde{v}_i(\omega(s_i)-\omega(s_{j_i}))} \Bigg ] d|\mu|(\xi_i)d|\mu|(\tilde{\xi}_i)d|\mu|(\zeta_i)d|\mu|(\tilde{\zeta}_i)du_id\tilde{u}_idv_id\tilde{v}_i \Bigg ).
\end{gather*}
Since by assumption the support of the measure $\mu$ is contained in a ball  $B_R$ of radius $R$, we can bound $|\xi_i||\tilde{\xi}_i||\zeta_i||\tilde{\zeta}_i|\le R^{4}$ on the support of $\mu$ obtaining :
\begin{gather*}
 \left ( I_n^2 \right )^2 \le R^{4p} {\sum_{j_1,\dots, j_p=0}^{n-1}} \left ( {\int_{s_{j_1}}^{s_{j_1+1}}}\cdots {\int_{s_{j_p}}^{s_{j_p+1}}} \int_0^1 \cdots \int_0^1 \int_{\R} \cdots \int_{\R} \E \left [  \prod_{i=1}^p e^{-\frac{\sqrt{2}}{2}\omega(s_i)(\xi_i+\tilde{\xi}_i+\zeta_i+\tilde{\zeta}_i)}\right . \right .\\
e^{-\frac{\sqrt{2}}{2}(\omega(s_i)-\omega(s_{j_i})(\xi_i u_i+\tilde{\xi}_i \tilde{u}_i + \zeta_i v_i + \tilde{\zeta}_i \tilde{v}_i)} \Bigg ] d|\mu|(\zeta_i)d|\mu|(\tilde{\zeta}_i)du_id\tilde{u}_i dv_i d\tilde{v}_i \Bigg ).
\end{gather*}
We notice that the term under the expectation can be computed  as
\[
\exp \left [ P(s_1, s_{j_1},\dots,s_p,s_{j_p},\xi_1,\tilde{\xi}_1 , \dots, v_p \tilde{v}_{p})\right ],
\]
where $P$ is a polynomial function, which maximum $M_P$ for $s_i,s_{k_i} \in [0,t]$, $u_i,\tilde{u}_i,v_i,\tilde{v}_i \in [0,1]$, and $\xi_i, \tilde{\xi}_i, \zeta_i, \tilde{\zeta}_i \in \supp(\mu)$, for all $i = 1 \dots p$.
Finally, by integrating and summing with respects to all variables, we get a finite term of the order $t^p\cdot |\mu|^{4p}\cdot M$, proving a uniform bound for  $I_n^2$. Hence\[
g_n(\omega) \xrightarrow{L^{2p}(C_t,\mathbb{P})} \int_0^t a(\sqrt{i\hbar}\omega(s))d\omega(s), \qquad \omega \in C_t.
\]

Let us consider now the sequence of random variables $\{h_n\}$ given by
$$ 
h_n(\omega)=\sum_{j=0}^{n-1} \frac{1}{2}\cdot a'( \sqrt{ i\hbar} {\omega(s_j)})({\omega(s_j)}j-{\omega(s_j)})^2, \qquad \omega \in C_t,
$$ 
and set $a'(\sqrt{i\hbar}\omega(s))\equiv \phi(\omega(s))$, for any $s \in [0,t$]. Let $H$ be the random variable defined by 
$$H(\omega) =\frac{1}{2} \int_0^t  \phi(\omega(s)) ds, \qquad \omega \in C_t.$$
We have:
\begin{align*}
H(\omega)-h_n(\omega)=&
\frac{1}{2}\sum_{j=0}^{n-1}\left ( \int_{s_j}^{s_{j+1}} \phi(\omega(s)) ds - \phi (\omega(s_j))\left (\omega(s_{j+1})-\omega(s_j)\right )^2\right )\\
=&\frac{1}{2}\sum_{j=0}^{n-1}\Bigg( \int_{s_j}^{s_{j+1}} \phi(\omega(s_j))ds + \int_{s_j}^{s_{j+1}}\phi'(\omega(s_j))(\omega(s)-\omega(s_j))ds\\
&+\int_{s_j}^{s_{j+1}} \int_0^1\left(\omega(s)-\omega(s_j))^2 \phi''(\omega (s_j)+u(\omega(s)-\omega(s_j)))(1-u)\right) duds \\
&\quad \qquad \qquad\qquad \qquad   - \phi (\omega(s_j)) \left (\omega(s_{j+1})-\omega(s_j) \right )^2 \Bigg ).
\end{align*}
Hence 
$$\|H-h_n\|_{L^{p}(C_t, \bP)}\leq \frac{1}{2}\left(\|J_n^1\|_{L^{p}(C_t, \bP)}+\|J_n^2\|_{L^{2p}(C_t, \bP)}+\|J_n^3\|_{L^{p}(C_t, \bP)}\right),$$
where:
\begin{eqnarray*}
J^1_n(\omega)&=& \sum_{j=0}^{n-1}  \phi(\omega(s_j)) \left((s_{j+1}-s_j)-(\omega(s_{j+1})-\omega(s_j))^2 \right ); \\
J_n^2 (\omega)&=&\sum_{j=0}^{n-1}  \phi'(\omega(s_j)) \int_{s_j}^{s_{j+1}} (\omega(s)-\omega(s_j))ds;  \\
J^3_n(\omega)&=& \sum_{j=0}^{n-1}\int_{s_j}^{s_{j+1}} \int_0^1\left(\omega(s)-\omega(s_j))^2 \phi''(\omega (s_j)+u(\omega(s)-\omega(s_j)))(1-u)\right) duds.
\end{eqnarray*}
Without loss of generality we can consider the case where the function $\phi:\R\to \C$ is real valued, since the general case follows easily by the inequality $\|J_n^1\|_{L^{p}}\leq \|Re(J_n^1)\|_{L^{p}}+\|Im(J_n^1)\|_{L^{p}}$.\\
The $L^{p}$ norm of the function  $J_n^1$ can be estimated as:
\begin{eqnarray*}
\bE[|J^1_n|^{2p}]&=&\sum_{j_1,...,j_{2p}=0}^{n-1}\mathbb{E}\Bigg [ \phi(\omega(s_{j_1}))\cdots  \phi(\omega(s_{j_{2p}}))\left((s_{{j_1}+1}-s_{j_1})-(\omega(s_{j_1+1})-\omega(s_{j_1}))^2 \right )\\
& &\qquad \cdots \left((s_{{j_{2p}}+1}-s_{j_{2p}})-(\omega(s_{j_{2p}+1})-\omega(s_{j_{2p}}))^2 \right )\Bigg ]\\ 
&\leq &(2p)!\sum_{0\leq j_1\leq...\leq j_{2p}\leq  n-1}\mathbb{E}\Bigg [ \phi(\omega(s_{j_1}))\cdots  \phi(\omega(s_{j_{2p}}))\left((s_{{j_1}+1}-s_{j_1})-(\omega(s_{j_1+1})-\omega(s_{j_1}))^2 \right )\\
& &\qquad \cdots \left((s_{{j_{2p}}+1}-s_{j_{2p}})-(\omega(s_{j_{2p+1}})-\omega(s_{j_{2p}}))^2 \right )\Bigg ]\ .
\end{eqnarray*}
Since $\bE[((s_{j+1}-s_{j})-(\omega(s_{j}+1)-\omega(s_j))^2 )]=0$, the sum above contains only the $n^{2p-1}$terms where $ j_1\leq\dots\leq j_{2p-1}=j_{2p}$. Indeed, if   $ j_1\leq\dots\leq j_{2p-1}<j_{2p}$:
\begin{gather*}
\mathbb{E}\Bigg [\prod_{i=1}^{2p} \phi(\omega(s_{j_i}))\left((s_{{j_1}+1}-s_{j_1})-(\omega(s_{j_1+1})-\omega(s_j))^2 \right )\Bigg]=\\
=\mathbb{E}\Bigg [\prod_{i=1}^{2p-1} \phi(\omega(s_{j_i}))\left((s_{{j_1}+1}-s_{j_1})-(\omega(s_{j_1+1})-\omega(s_j))^2 \right )\phi(\omega(s_{j_{p}}))\Bigg] \cdot \\
\mathbb{E}\Bigg [\left((s_{{j_{2p}}+1}-s_{j_{2p}})-(\omega(s_{j_{2p+1}})-\omega(s_{j_{2p}}))^2 \right )\Bigg]=0.
\end{gather*}
Direct computation shows that all the terms in this sum are of order $ O((s_{{j}+1}-s_{j})^{2p})=O(1/n^{p})$ or less.  
 Indeed, taking into account the possible coincidences of indexes, all the terms are of the form
 \begin{multline}\label{nnpp}
 \mathbb{E}\Bigg [ \phi(\omega(s_{k_1}))^{p_1}\left((s_{{k_1}+1}-s_{k_1})-(\omega(s_{k_1+1})-\omega(s_{k_1}))^2 \right )^{p_1}\cdots \\ \cdots \phi(\omega(s_{k_r}))^{p_r}\left((s_{{k_r}+1}-s_{k_r})-(\omega(s_{k_r+1})-\omega(s_{k_r}))^2 \right )^{p_r}\Bigg],
\end{multline}
where $p_1+\dots+p_r=2p$ and $k_1<k_2< \cdots <k_r$. By writing $\phi(x)=\int e^{i\sqrt i\xi x}
d\nu(\xi)$, $x \in R$ with $\nu $ complex Borel measure on $\R$ supported in the ball $B_R$, the integral \eqref{nnpp} can be estimates as:
\begin{gather}
\int_{\R^{2p}}\bE\left[\left((s_{{k_r}+1}-s_{k_r})-(\omega(s_{k_r+1})-\omega(s_{k_r}))^2 \right )^{p_r}\right]\prod_{\alpha=0}^{r-2}\Bigg(\bE\left[ e^{i\sqrt i(\omega (s_{k_{r-\alpha}})-\omega (s_{k_{r-\alpha-1}+1}))\sum_{l=1}^{\sum_{\beta=0}^\alpha p_{r-\beta}}\xi_l}\right]\Bigg) \nonumber \\
\Bigg(\prod_{\alpha =1}^{r-1}\bE\left[ e^{i\sqrt i(\omega (s_{k_{r-\alpha}+1})-\omega (s_{k_{r-\alpha}}))\sum_{l=1}^{\sum_{\beta=0}^\alpha p_{r-\beta}}\xi_l}\left((s_{{k_{r-\alpha}}+1}-s_{k_{r-\alpha}})-(\omega(s_{k_{r-\alpha}+1})-\omega(s_{k_{r-\alpha}}))^2 \right )^{p_{r-\alpha}}\right]\Bigg) \nonumber \\
\bE\left[e^{i\sqrt i\omega (s_{k_1})\sum_{l=1}^{2p}\xi_l}\right]]d\nu (\xi_1)\dots d\mu (\xi_{2p}).
\end{gather}
Now, since $\omega(t_1)-\omega(t_2)$ has the same law as $(t_1-t_2)^{\frac{1}{2}}X$, with $X$ a standard normal random variable and for all $\zeta \in \R$, $0\leq t_1\leq t_2$, $k\in \N$, we have:
\begin{eqnarray*}&&\bE[e^{i\sqrt i \zeta (\omega(t_1)-\omega(t_2)) }]=e^{-\frac{i}{2}(t-s)\xi ^2},\\
& &\bE[e^{i\sqrt i \zeta X }X^{2k}]=H_{2k}(\sqrt i \zeta)e^{-\frac{i}{2}\zeta^2},
\end{eqnarray*}
with $H_n$ denoting the $n-th$ Hermite polynomial.
Hence:
\begin{eqnarray*}
&&\Bigg | \bE\left[ e^{i\sqrt i(\omega (s_{k_{r-\alpha}})-\omega (s_{k_{r-\alpha-1}+1})) \sum_{l=1}^{\sum_{\beta=0}^\alpha p_{r-\beta}}\xi_l}\right]\Bigg|=1\\
&&\Bigg|\bE\left[ e^{i\sqrt i(\omega (s_{k_{r-\alpha}+1})-\omega (s_{k_{r-\alpha}}))\sum_{l=1}^{\sum_{\beta=0}^\alpha p_{r-\beta}}\xi_l}\left((s_{{k_{r-\alpha}}+1}-s_{k_{r-\alpha}})-(\omega(s_{k_{r-\alpha}+1})-\omega(s_{k_{r-\alpha}}))^2 \right )^{p_{r-\alpha}}\right]\Bigg|\leq \\ & & \qquad \qquad  (s_{{k_{r-\alpha}}+1}-s_{k_{r-\alpha}})^{p_{r-\alpha}}P_{\alpha ,p_{r-\alpha}}(\xi_1, \dots,\xi_{2p}),
\end{eqnarray*}
with $P_{\alpha ,p_{r-\alpha}}:\R^{2p}\to \R$ suitable polynomial functions.
By setting 
$$
M:=\max_{r, p_1, \dots p_r}\prod_{\alpha=1}^{r-1}\max_{\xi_1,\dots\xi _p\in B_R}|P_{\alpha ,p_{r-\alpha}}(\xi_1, \dots,\xi_p)|,
$$
we get
$\bE[|J^1_n|^{2p}]\leq M\frac{t^{2p}}{n}|\nu(B_R)|^{2p}$, obtaining the required convergence result:
 $$\lim_{n\to \infty}\bE[|J^1_n|^{2p}]\to 0.$$
The same argument produces an  analogous estimate for $\bE[|J^2_n|^{2p}]$. Indeed, always assuming without loss of generality that  the function $\phi$ is real valued, we get:
\begin{eqnarray*}
\bE[|J^2_n|^{2p}]&=&
\sum_{j_1,\dots,j_{2p} =0}^{n-1}\mathbb{E}\Bigg [ \phi'(\omega(s_{j_1}))\cdots  \phi'(\omega(s_{j_{2p}}))
\int_{s_{j_1}}^{s_{{j_1}+1}} (\omega(u_1)-\omega(s_{j_1}))du_1\\ & &\qquad \qquad \qquad \qquad \cdots \int_{s_{j_{2p}}}^{s_{{j_{2p}}+1}} (\omega(u_{2p})-\omega(s_{j_{2p}}))du_{2p}\\
&\leq &(2p)!\sum_{0\leq j_1\leq\dots\leq j_{2p}\leq  n-1}^{n-1}\mathbb{E}\Bigg [ \phi'(\omega(s_{j_1}))\cdots  \phi'(\omega(s_{j_{2p}}))
\int_{s_{j_1}}^{s_{{j_1}+1}} (\omega(u_1)-\omega(s_{j_1}))du_1\\ & &\qquad \qquad \qquad \qquad \cdots \int_{s_{j_{2p}}}^{s_{{j_{2p}}+1}} (\omega(u_{2p})-\omega(s_{j_{2p}}))du_{2p}.\end{eqnarray*}
Again, since $\bE[\int_{s_j}^{s_j+1}(\omega(u)-\omega (s_j))du]=0$, we can consider only the $n^{2p-1} $ terms with $ j_1\leq \cdots\leq  j_{2p-1}= j_{2p}$.  All terms  have the same structure as the integrals appearing in \eqref{69bis} and by using  the same arguments applied for the estimates of integrals \eqref{I1n} and \eqref{I2n}, we  obtain $\lim_{n\to \infty }\bE[|J^2_n|^{2p}]=0$.
Furthermore, the same argument applies also to the term $J_n^3$, yielding $\lim_{n\to \infty }\bE[|J^3_n|^{2p}]=0$.

Thus
\[
h_n \xrightarrow{L^{p}(\Omega,\mathbb{P})} \int_{0}^t \phi(\omega(s))ds.
\]

We estimate the last term $r_n$ by the Cauchy-Schwarz inequality as follows
\begin{align}\nonumber
|r_n|^{2p}\le& \E \left [ \left ( \sum_{j=0}^{n-1} \int_0^1\int_0^1\int_{\R}\int_{\R} |\kappa_1||\kappa_2| e^{-\frac{\sqrt{2}}{2} \kappa_1 ({\omega(s_j)}+(\omega(s_{j+1})-\omega(s_j))u_1)} e^{-\frac{\sqrt{2}}{2} \kappa_2 ({\omega(s_j)}+(\omega(s_{j+1})-\omega(s_j))u_2)} \right . \right . \nonumber\\
& \qquad ({\omega(s_j)}j-{\omega(s_j)})^6 (1-u_1)^2(1-u_2)^2 du_1du_2d|\mu|(\kappa_1)d|\mu|(\kappa_2) \bigg)^p \vast ]\nonumber\\
\le& \sqrt{\E \left [ {\sum_{j_1,\dots, j_p=0}^{n-1}} (\omega(s_{j_1+1})-\omega(s_{j_1}))^{12}\cdots (\omega(s_{j_p+1})-\omega(s_{j_p}))^{12}\right ]}\nonumber\\
& \ \sqrt{\E \left [ {\sum_{j_1,\dots, j_p=0}^{n-1}} \mathcal{F}(\omega(s_{j_1+1}),\omega(s_{j_1}))^2 \cdots \mathcal{F}(\omega(s_{j_p}),\omega(s_{j_p+1}))^2\right ] },\label{eqnnn}
\end{align}
where 
\begin{align*}
\mathcal{F}({\omega(s_j)},{\omega(s_j)}j)= &\int_0^1\int_0^1 \int_{\R}\int_{\R}|\kappa_1||\kappa_2|  e^{-\frac{\sqrt{2}}{2} \kappa_1 ({\omega(s_j)}+(\omega(s_{j+1})-\omega(s_j))u_1)} \\
&e^{-\frac{\sqrt{2}}{2} \kappa_2 ({\omega(s_j)}+(\omega(s_{j+1})-\omega(s_j))u_2)} (1-u_1)^2(1-u_2)^2 du_1du_2d|\mu|(\kappa_1)d|\mu|(\kappa_2).
\end{align*}
Both factors appearing in the last line of  \eqref{eqnnn}
 can be estimated by the same techniques applied in the study of the terms \eqref{I1n} and \eqref{I2n},  
obtaining
$r_n \xrightarrow{L^p(\Omega,\mathbb{P})}0$.

Eventually, we conclude that the sequence of random variables $f_n$ defined as 
\[f_n(\omega)=\int_0^t \ba \left (  \sqrt{ i\hbar} \omega_n(s) \right )\cdot \dot{\omega}_n(s) ds, 
\]
converges, as $n \to \infty$, in $L^p(\Omega,\mathbb{P})$ to the random variable $f$ defined as the Stratonovich stochastic integral 
\[
f(\omega)=\int_0^t \ba ( \sqrt{ i\hbar} \omega(s))\circ d\omega(s).
\]
\finedim

\end{document}